\documentclass[a4paper,twocolumn,superscriptaddress,11pt,accepted=2019-04-25]{quantumarticle}
\pdfoutput=1
\usepackage{amsfonts}
\usepackage{enumitem}
\usepackage{amssymb}
\usepackage{extarrows}
\usepackage{amsmath}
\usepackage{amstext}
\usepackage{graphicx}
\usepackage{latexsym}
\usepackage{amsthm}
\usepackage{mathtools}
\usepackage{amsthm}
\usepackage{tikz}
\usepackage{tabularx}
	\newcolumntype{L}{>{\raggedright\arraybackslash}X}
\usepackage{mathrsfs} 

\def\doint{\mathop{\displaystyle \oint}}

\makeatletter
  \newcommand{\displaybump}{\hbox to \@totalleftmargin{\hfil}}
\makeatother

\begin{document}

\newtheorem{theorem}{Theorem}
\newtheorem{acknowledgement}[theorem]{Acknowledgment}
\newtheorem{algorithm}[theorem]{Algorithm}
\newtheorem{axiom}{Axiom}
\newtheorem{claim}[theorem]{Claim}
\newtheorem{conclusion}[theorem]{Conclusion}
\newtheorem{condition}[theorem]{Condition}
\newtheorem{conjecture}[theorem]{Conjecture}
\newtheorem{corollary}[theorem]{Corollary}
\newtheorem{criterion}[theorem]{Criterion}
\newtheorem{definition}[theorem]{Definition}
\newtheorem{example}[theorem]{Example}
\newtheorem{exercise}[theorem]{Exercise}
\newtheorem{lemma}[theorem]{Lemma}
\newtheorem{notation}[theorem]{Notation}
\newtheorem{problem}[theorem]{Problem}
\newtheorem{proposition}[theorem]{Proposition}
\newtheorem{remark}[theorem]{Remark}
\newtheorem{solution}[theorem]{Solution}
\newtheorem{summary}[theorem]{Summary}
\newtheorem{prop}[theorem]{Proposition}
\newtheorem{request}{Request}

\title{Infinitesimal and Infinite Numbers as an Approach to Quantum Mechanics}

\author{Vieri Benci}
\email{vieri.benci@unipi.it}
\orcid{0000-0003-0454-0939}
\affiliation{Dipartimento di Matematica, Universit\`{a} degli Studi di Pisa, Via F. Buonarroti 1/c, 56127 Pisa, Italy}
\author{Lorenzo Luperi Baglini}
\email{lorenzo.luperi.baglini@univie.ac.at}
\orcid{0000-0002-0559-0770}
\affiliation{Fakult\"{a}t f\"{u}r Mathematik, Universit\"{a}t Wien, Oskar-Morgenstern Platz 1, 1090 Vienna, Austria}
\author{Kyrylo Simonov}
\email{kyrylo.simonov@univie.ac.at}
\orcid{0000-0001-6764-8555}
\affiliation{Fakult\"{a}t f\"{u}r Mathematik, Universit\"{a}t Wien, Oskar-Morgenstern Platz 1, 1090 Vienna, Austria}

\maketitle

\begin{abstract}
Non-Archimedean mathematics is an approach based on fields which contain infinitesimal and infinite elements. Within this approach, we construct a space of a particular class of generalized functions, ultrafunctions. The space of ultrafunctions can be used as a richer framework for a description of a physical system in quantum mechanics. In this paper, we provide a discussion of the space of ultrafunctions and its advantages in the applications of quantum mechanics, particularly for the Schr\"{o}dinger equation for a Hamiltonian with the delta function potential.
\end{abstract}


\section{Introduction\label{intro}}

Quantum mechanics is a highly successful physical theory, which provides a counter-intuitive but accurate description of our world. During more than 80 years of its history, there were developed various formalisms of quantum mechanics that use the mathematical notions of different complexity to derive its basic principles. The standard approach to quantum mechanics handles linear operators, representing the observables of the quantum system, that act on the vectors of a Hilbert space representing the physical states. However, the existing formalisms include not only the standard approach but as well some more abstract approaches that go beyond Hilbert space. A notable example of such an abstract approach is the algebraic quantum mechanics, which considers the observables of the quantum system as a non-abelian $C^*$-algebra, and the physical states as positive functionals on it~\cite{Strocchi}.

Non-Archimedean mathematics (particularly, nonstandard analysis) is a framework that treats the infinitesimal and infinite quantities as numbers. Since the introduction of nonstandard analysis by Robinson~\cite{Robinson} non-Archimedean mathematics has found a plethora of applications in physics~\cite{Arkeryd, Albeverio1988, Albeverio1986, Harthong1981, Harthong1984}, particularly in quantum mechanical problems with singular potentials, such as $\delta$-potentials~\cite{Farrukh, Albeverio1979, Bagarello, AlbeverioBook, Raab}.

In this paper, we build a non-Archimedean approach to quantum mechanics in a simpler way through a new space, which can be used as a basic construction in the description of a physical system, by analogy with the Hilbert space in the standard approach. This space is called the space of ultrafunctions, a particular class of non-Archimedean generalized functions~\cite{ultra, milano, algebra, belu2012, belu2013, BLS, beyond}. The ultrafunctions are defined on the hyperreal field $\mathbb{R}^*$, which extends the reals $\mathbb{R}$ by including infinitesimal and infinite elements into it. Such a construction allows studying the problems which are difficult to solve and formalize within the standard approach. For example, variational problems, that have no solutions in standard analysis, can be solved in the space of ultrafunctions~\cite{BLS}. In this way, non-Archimedean mathematics as a whole and the ultrafunctions as a particular propose a richer framework, which highlights the notions hidden in the standard approach and paves the way to a better understanding of quantum mechanics.

The paper is organized as follows. In Section~\ref{intro} we introduce the needed notations and the notion of a non-Archimedean field. In Section~\ref{lt} we introduce a particular non-Archimedean field, the field of Euclidean numbers $\mathbb{E}$, through the notion of $\Lambda$-limit, which is a useful, straightforward approach to the nonstandard analysis. In Section~\ref{Ultrafunctions} we introduce the space of ultrafunctions, which are a particular class of generalized functions. In Section~\ref{QM} we apply the ultrafunctions approach to quantum mechanics and discuss its advantages in contrast to the standard approach. In Section~\ref{deltaP} we provide a discussion of a quantum system with a delta function potential for the standard and ultrafunctions approaches. Last but not least, in Section~\ref{concl} we provide the conclusions.

\subsection{Notations\label{not}}

Let $\Omega $\ be an open subset of $\mathbb{R}^{N}$, then

\begin{itemize}
\item $C\left( \Omega \right) $ denotes the set of continuous functions defined on $\Omega \subset \mathbb{R}^{N}$,
\item $\mathcal{C}_{c}\left( \Omega \right) $ denotes the set of continuous functions in $C\left( \Omega \right) $ having compact support in $\Omega$,
\item $\mathcal{C}^{k}\left( \Omega \right) $ denotes the set of functions defined on $\Omega \subset \mathbb{R}^{N}$ which have continuous derivatives up to the order $k$,
\item $\mathcal{C}_{c}^{k}\left( \Omega \right) $ denotes the set of functions in $C^{k}\left( \Omega \right) \ $having compact support,
\item $\mathcal{D}\left( \Omega \right) $ denotes the space of infinitely differentiable functions with compact support defined almost everywhere in $\Omega$,
\item $L^{2}\left( \Omega \right) $ denotes the space of square integrable functions defined almost everywhere in $\Omega$,
\item $\mathcal{L}^{1}_{loc}\left( \Omega \right) $ denotes the space of locally integrable functions defined almost everywhere in $\Omega$,
\item $\mathfrak{mon}(x)=\{y\in \mathbb{E}^{N}\ |\ x\sim y\}\ $ (see Def. \ref{MG}),
\item given any set $E\subset X$, $\chi _{E}:X\rightarrow \mathbb{R}$ denotes the characteristic function of $E$, namely,
\begin{equation*}
\chi _{E}(x):=\left\{ 
\begin{array}{cc}
1 & \text{if}\ \ x\in E, \\ 
&  \\ 
0 & \text{if}\ \ x\notin E,
\end{array}
\right.
\end{equation*}
\item with some abuse of notation, we set $\chi _{a}(x):=\chi _{\left\{a\right\} }(x)$,
\item $\partial _{i}=\frac{\partial }{\partial x_{i}}$ denotes the usual partial derivative, $D_{i}$ denotes the generalized derivative (see Section \ref{adu}),
\item $\int $ denotes the usual Lebesgue integral, $\oint $ denotes the pointwise integral (see Section \ref{adu}),
\item if $E$ is any set, then $|E|$ denotes its cardinality.
\end{itemize}

\subsection{Non-Archimedean fields\label{naf}}

Our approach to quantum mechanics makes multiple uses of the notions of infinite and infinitesimal numbers. A natural framework to introduce these numbers suitably is provided by non-Archimedean mathematics (see, e.g., \cite{el06}). This framework operates with the infinite and infinitesimal numbers as the elements of the new \textit{non-Archimedean field}.

\begin{definition}
Let $\mathbb{K}$ be an infinite ordered field\footnote{Without loss of generality, we assume that $\mathbb{Q}\subseteq\mathbb{K}$.}. An element $\xi \in \mathbb{K}$ is:

\begin{itemize}
\item infinitesimal if, for all positive $n\in \mathbb{N}$, $|\xi|<\frac{1}{n}$,
\item finite if there exists $n\in \mathbb{N}$ such that $|\xi |<n$,
\item infinite if, for all $n\in \mathbb{N}$, $|\xi |>n$ (equivalently, if $\xi $ is not finite).
\end{itemize}
We say that $\mathbb{K}$ is non-Archimedean if it contains an infinitesimal $\xi \neq 0$, and that $\mathbb{K}$ is superreal if it properly extends $\mathbb{R}$.
\end{definition}

Notice that, trivially, every superreal field is non-Archimedean. Infinitesimals allow introducing the following equivalence relation, which is fundamental in all non-Archimedean settings.
\begin{definition}
\label{def infinite closeness} We say that two numbers $\xi ,\zeta \in {\mathbb{K}}$ are infinitely close if $\xi -\zeta $ is infinitesimal. In this case we write $\xi \sim \zeta $.
\end{definition}

In the superreal case, $\sim$ can be used to introduce the fundamental notion of ``standard part''\footnote{For a proof of the following simple theorem, the interested reader can check, e.g., \cite{Goldblatt}.}.
\begin{theorem}
If $\mathbb{K}$ is a superreal field, every finite number $\xi \in \mathbb{K}$ is infinitely close to a unique real number $r\sim \xi $, called the the \textbf{standard part} of $\xi $.
\end{theorem}

Following the literature, we will always denote by $\operatorname{st}(\xi )$ the standard part of any finite number $\xi$. Moreover, with a small abuse of notation, we also put $\operatorname{st}(\xi ) = +\infty $ (resp. $\operatorname{st}(\xi )=-\infty $) if $\xi \in \mathbb{K}$ is a positive (resp. negative) infinite number.
\begin{definition}
\label{MG}Let $\mathbb{K}$ be a superreal field, and $\xi \in \mathbb{K}$ a number. The \label{def monad} monad of $\xi $ is the set of all numbers that are infinitely close to it,
\begin{equation*}
\mathfrak{m}\mathfrak{o}\mathfrak{n}(\xi) = \{\zeta \in \mathbb{K}:\xi \sim \zeta \}.
\end{equation*}
\end{definition}

Notice that, by definition, the set of infinitesimals is $\mathfrak{mon}(0)$ precisely. Finally, superreal fields can be easily complexified by considering
\begin{equation*}
\mathbb{K} + i\mathbb{K},
\end{equation*}
namely, a field of numbers of the form 
\begin{equation*}
a+ib,\ a,b\in \mathbb{K}.
\end{equation*}
In this way, the complexification of non-Archimedean fields shows no particular difficulty and is straightforward.

\bigskip

\section{The field of Euclidean numbers\label{lt}}

Nonstandard analysis plays one of the most prominent roles between various approaches to non-Archimedean mathematics. One reason is that nonstandard analysis provides a handy tool to study and model the problems which come from many different areas. However, the classical representations of nonstandard analysis can feel overwhelming sometimes, as they require a good knowledge of the objects and methods of mathematical logic. This stands in contrast to the actual use of nonstandard objects in the mathematical practice, which is almost always extremely close to the usual mathematical practice. 

For these reasons, we believe that it is worth to present nonstandard analysis avoiding most of the usual logic machinery. We will introduce the nonstandard approach to the formalism of quantum mechanics via the notion of $\Lambda$-limit\footnote{In \cite{fle} and \cite{BDNF2}, the reader can find several other approaches to nonstandard analysis and an analysis of them.} and the Euclidean numbers. These numbers are the underlying object of $\Lambda $-theory, and can be introduced via a purely algebraic approach, as done in \cite{benci99}\footnote{An elementary presentation of (part of) this theory can be found in \cite{Bmathesis} and \cite{bencilibro}.}. 

The basic idea of the construction of Euclidean numbers is the following: as the real numbers can be constructed by completion of the rational numbers using the Cauchy notion of limit, the Euclidean numbers can be constructed by completion of the reals with respect to a new notion of limit, called $\Lambda$-limit, that we are now going to introduce.

Let $\Lambda $ be an infinite set containing $\mathbb{R}$ and let $\mathfrak{L}$ be the family of finite subsets of $\Lambda$. A function $\varphi : \mathfrak{L}\rightarrow \mathbb{R}$ will be called \textit{net} (with values in $\mathbb{R}$). The set of such nets is denoted by $\mathfrak{F}\left( \mathfrak{L},\mathbb{R}\right)$ and equipped with the natural operations
\begin{eqnarray*}
\left( \varphi +\psi \right) (\lambda ) &=& \varphi(\lambda) + \psi(\lambda), \\
\left( \varphi \cdot \psi \right) (\lambda ) &=& \varphi(\lambda) \cdot \psi(\lambda),
\end{eqnarray*}
and the partial order relation
\begin{equation*}
\varphi \geq \psi \; \Longleftrightarrow \; \forall \lambda \in \mathfrak{L} ,\ \varphi(\lambda) \geq \psi(\lambda).
\end{equation*}
In this way, $\mathfrak{F}\left( \mathfrak{L},\mathbb{R}\right)$ is a partially ordered real algebra.

\begin{definition}
We say that a superreal field $\mathbb{E}$ is a field of Euclidean numbers if there is a surjective map 
\begin{equation*}
J:\mathfrak{F}\left( \mathfrak{L},\mathbb{R}\right) \rightarrow \mathbb{E},
\end{equation*}
which satisfies the following properties,
\begin{itemize}
\item $J\left( \varphi + \psi \right) = J \left(\varphi \right) + J\left(\psi\right)$,

\item $J\left( \varphi \cdot \psi \right) = J\left( \varphi \right) \cdot J\left( \psi \right)$,

\item if $\varphi (\lambda ) > r$, then $J\left( \varphi \right) > r$.
\end{itemize}
$J$ will be called the realization of $\mathbb{E}$.
\end{definition}

The proof of the existence of such a field is an easy consequence of the Krull -- Zorn theorem. It can be found, e.g., in \cite{benci99,ultra,milano,bencilibro}. In this paper, we also use the complexification of $\mathbb{E}$, denoted\footnote{The choice of the notation will become clear later on.} by
\begin{equation}
\mathbb{C}^{\ast} = \mathbb{E} + i\mathbb{E}.  \label{rita}
\end{equation}

\begin{definition} Let $\mathbb{E}$ be a field of Euclidean numbers, let $J$ be its realization and let $\varphi\in\mathfrak{F}\left(\mathfrak{L},\mathbb{R}\right)$. We say that $J\left(\varphi\right)$ is the \textbf{$\Lambda$-limit} of the net $\varphi$. We will also denote it by 
\begin{equation*}
J\left( \varphi \right) := \lim_{\lambda \uparrow \Lambda }\varphi (\lambda ).
\end{equation*}

\end{definition}

The name ``limit'' has been chosen because the operation
\begin{equation*}
\varphi \mapsto \lim_{\lambda \uparrow \Lambda }\varphi (\lambda )
\end{equation*}
satisfies the following properties,

\begin{itemize}
\item \textsf{(}$\Lambda $-\textsf{1)}\ \textbf{Existence.}\ \textit{Every net $\varphi :\mathfrak{L}\rightarrow \mathbb{R}$ has a unique limit $L\in \mathbb{E}$.}

\item ($\Lambda $-2)\ \textbf{Constant}. \textit{If $\varphi(\lambda)$ is eventually constant, namely, $\exists \lambda_0\in \mathfrak{L}$, $r\in \mathbb{R}$ such that $\forall \lambda \supset \lambda_0$, $\varphi (\lambda )=r$, then}
\begin{equation*}
\lim_{\lambda \uparrow \Lambda} \varphi(\lambda) = r.
\end{equation*}

\item ($\Lambda $-3)\ \textbf{Sum and product}.\ $\forall \varphi
,\psi :\mathfrak{L}\rightarrow \mathbb{R}$
\begin{alignat*}{5}
\nonumber &\lim\limits_{\lambda \uparrow \Lambda} \varphi(\lambda) &\hspace{1mm}+\hspace{0.5mm}& \lim\limits_{\lambda \uparrow \Lambda} \psi(\lambda) &\hspace{0.5mm}\;=\;& \lim\limits_{\lambda \uparrow \Lambda} (\varphi(\lambda) &\hspace{0.5mm}+\hspace{0.5mm}& \psi(\lambda)), \\
\nonumber &\lim\limits_{\lambda \uparrow \Lambda} \varphi(\lambda) &\cdot \hspace{1.5mm} & \lim\limits_{\lambda \uparrow \Lambda} \psi (\lambda) &\hspace{0.5mm}\;=\;&  \lim\limits_{\lambda \uparrow \Lambda} (\varphi(\lambda) &\cdot \hspace{1.5mm}& \psi(\lambda)).
\end{alignat*}
\end{itemize}

The Cauchy limit of a sequence (as formalized by Weierstra\ss) satisfies the second and the third\footnote{When both sequences have a limit.} condition above. The main difference between the Cauchy and the $\Lambda$-limit is given by the property $(\Lambda-1)$, namely, by the fact that the $\Lambda$-limit always exists. Notice that it implies that $\mathbb{E}$ must be larger than $\mathbb{R}$ since it must contain the limit of diverging nets as well. 

For those nets that have a Cauchy limit, the relationship between the Cauchy and the $\Lambda$-limit is expressed by the following identity,
\begin{equation}
\lim_{\lambda \rightarrow \Lambda }\varphi (\lambda )=\operatorname{st}\left( \lim_{\lambda
\uparrow \Lambda }\varphi (\lambda )\right).  \label{bn}
\end{equation}

With the introduction of the Euclidean numbers there arises a natural question: What do they look like? Let us give a few examples.
\begin{enumerate}
    \item Let $\varphi(\lambda):=\frac{1}{|\lambda|}$ for every $\lambda\neq\emptyset$, and $\varepsilon=\lim_{\lambda \uparrow \Lambda }\varphi (\lambda )$. Then $\varepsilon>0$, because $\varphi(\lambda)>0$ for every $\lambda\in\mathfrak{L}$. However, $\varepsilon<r$ for every positive $r\in\mathbb{R}$ since $\varphi(\lambda)<r$ eventually in $\lambda$. Therefore, $\varepsilon$ is a non-zero infinitesimal in $\mathbb{E}$.
    
    \item Analogously, let $\varphi(\lambda):=|\lambda|$ for every $\lambda\in\mathfrak{L}$. Let $\sigma=\lim_{\lambda \uparrow \Lambda }\varphi (\lambda )$. Straightforwardly, $\sigma$ is a positive infinite element in $\mathbb{E}$, and $\sigma\cdot\varepsilon=1$.
    
    \item Let us generalize the previous example. If $E\subset \mathbb{R}\subset \Lambda,$ we set
\begin{equation*}
\mathfrak{n}\left( E\right) =\lim_{\lambda \uparrow \Lambda }\ |E\cap \lambda |,
\end{equation*}
where $|F|$ denotes the number of elements of a finite set $F$. Since $\lambda$ is a finite set, $|E\cap\lambda|\in\mathbb{N}$ for every $\lambda$. If $E$ is finite, then $E$ belongs to $\mathfrak{L}$, and the net $|E\cap\lambda|$ is eventually constant (in $\lambda$) and equal to $|E|$, so that $\mathfrak{n}\left(E\right)=|E|$. On the other hand, if $E$ is an infinite set, the above limit gives an infinite number (called the \textit{numerosity} of $E$\footnote{The reader interested in the details and the developments of the theory of numerosities is referred to \cite{benci95b,BDN2003,BDNF1,BF}.}).
\end{enumerate}

\begin{definition}
A mathematical entity is called \textbf{internal} if it is a $\Lambda$-limit of some other entities.
\end{definition}

\subsection{Extension of functions, hyperfinite sets and grid functions\label{HE}}

The $\Lambda$-limit allows extending the field of real numbers to the field of Euclidean numbers. Similarly, we can use it to extend sets and functions in arbitrary dimensions,
\begin{itemize}
    \item if $N\in\mathbb{N}$, then, in order to extend $\mathbb{R}^{N}$ to $\mathbb{E}^{N}$, we proceed in a trivial way: if $\varphi(\lambda):=\left( \varphi1
_{1}(\lambda ),...,\varphi _{N}(\lambda )\right) \in \mathbb{R}^{N},$ we set
\begin{equation*}
\lim\limits_{\lambda \uparrow \Lambda }\varphi (\lambda ) = \left( \lim\limits_{\lambda \uparrow \Lambda }~\varphi _{1}(\lambda ),...,\lim_{\lambda \uparrow \Lambda}~\varphi _{N}(\lambda )\right) \in \mathbb{E}^{N},
\end{equation*}

\item if $A\subset \mathbb{R}^{N}$, we let the \textbf{natural extension} of $A$ to be 
\begin{equation*}
A^{\ast } = \left\{ \lim\limits_{\lambda \uparrow \Lambda }\varphi (\lambda )\ |\ \forall \lambda ,\ \varphi (\lambda )\in A\right\},
\end{equation*}

\item if
\begin{equation*}
f: A\rightarrow \mathbb{R},\ \ A\subset \mathbb{R}^{N},
\end{equation*}
we let the \textbf{natural extension} of $f$ to $A^{\ast }$ to be a function such that, for every $x=\lim_{\lambda \uparrow \Lambda } x_{\lambda }\in A^{\ast}$,
\begin{equation*}
f^{\ast }\left( \lim_{\lambda \uparrow \Lambda }\ x_{\lambda }\right) = \lim\limits_{\lambda \uparrow \Lambda }~f^{\ast }\left( x_{\lambda} \right),
\end{equation*}

\item the above case can be generalized, in order to define directly the functions between subsets of $\mathbb{E}^{N}$, if
\begin{equation*}
u_{\lambda }:A\rightarrow \mathbb{R},\ \ A\subset \mathbb{R}^{N}
\end{equation*}
is a net of functions, its $\Lambda $-limit 
\begin{equation*}
u = \lim\limits_{\lambda \uparrow \Lambda }~u_{\lambda }: A^{\ast} \rightarrow \mathbb{E}
\end{equation*}
is a function such that, for any $x=\lim_{\lambda \uparrow \Lambda }\
x_{\lambda }\in \mathbb{E}^{N}$,  
\begin{equation*}
u(x)=\lim\limits_{\lambda \uparrow \Lambda }~u_{\lambda }\left( x_{\lambda }\right),
\end{equation*}

\item finally, if $V$ is a function space, we let its natural extension to be 
\begin{equation*} V^{\ast}:=\left\{\lim\limits_{\lambda\uparrow\Lambda} u_{\lambda}\mid u_{\lambda}\in V\right\}.
\end{equation*}

\end{itemize}
Notice that if, for all $x\in \mathbb{R}^{N}$,
\begin{equation*}
v(x)=\lim\limits_{\lambda \rightarrow \Lambda }~u_{\lambda } \left( x \right)
\end{equation*}
by Eq.~(\ref{bn}), it follows that
\begin{equation*}
\forall x\in \mathbb{R}^{N},\ \ v(x) = \operatorname{st}\left[ u(x)\right].
\end{equation*}

Moreover, it is clear that $\mathbb{E}$ is the natural extension of $\mathbb{R}$, namely, $\mathbb{E} = \mathbb{R}^{\ast}$. Notice that this choice justifies also our notation (\ref{rita}) for the complexification of $\mathbb{E}$,
\begin{equation*}
\mathbb{C}^{\ast } = \mathbb{E} + i\mathbb{E} = \mathbb{R}^{\ast} + i\mathbb{R}
^{\ast }.
\end{equation*}

Now we introduce a fundamental notion for all our applications, the \textit{``hyperfinite'' set}.
\begin{definition}
We say that a set $F\subset \mathbb{E}$ is \textbf{hyperfinite} if there is a net $\left\{ F_{\lambda }\right\} _{\lambda \in \Lambda }$ of finite sets such that 
\begin{equation*}
F = \left\{ \lim_{\lambda \uparrow \Lambda }\ x_{\lambda }\ |\ x_{\lambda } \in F_{\lambda }\right\}.
\end{equation*}
\end{definition}

Hyperfinite sets play the role of finite sets in the non-Archimedean framework since they share many properties with finite sets. We will often use the following feature provided by the hyperfinite sets: it is possible to ``add'' the elements of a hyperfinite set of numbers. If $F$ is a hyperfinite set of numbers, the \textbf{hyperfinite sum} of the elements of $F$ is defined in the following way,
\begin{equation*}
\sum_{x\in F}x=\ \lim\limits_{\lambda \uparrow \Lambda }\sum_{x\in F_{\lambda }}x.
\end{equation*}

Let us illustrate the hyperfinite sets with a very simple example. Let $F_{\lambda} = \{n\in\mathbb{N}\mid n\leq |\lambda|\}$ for every $\lambda\in\mathfrak{L}$. Further, let $F = \left\{ \lim_{\lambda \uparrow \Lambda }\ x_{\lambda }\ |\ x_{\lambda }\in F_{\lambda }\right\}$. Then $F$ is hyperfinite and $F\subseteq\mathbb{N}^{\ast}$. Moreover, $n\in F$, for every $n\in\mathbb{N}$, as $n\in F_{\lambda}$, for every $\lambda$, so that $|\lambda|\geq n$. Therefore, we have constructed a hyperfinite subset of $\mathbb{N}^{\ast}$ which contains the infinite set\footnote{At first sight, it might seem absurd to have a set that extends $\mathbb{N}$ but still behaves like a finite set. This is one of the peculiar and key properties of hyperfinite sets.} $\mathbb{N}$.

The kind of hyperfinite sets that we will use are the so-called ``hyperfinite grids''.
\begin{definition}
A hyperfinite set $\Gamma $ such that $\mathbb{R}^{N}\subset \Gamma \subset \mathbb{E}^{N}$ is called \textbf{hyperfinite grid}.
\end{definition}

If $\{\Gamma _{\lambda }\}_{\lambda\in\mathfrak{L}}$ is a family of finite subsets of $\mathbb{R}^N$, which satisfies the property
\begin{equation*}
\mathbb{R}^{N}\cap \lambda \subset \Gamma _{\lambda },
\end{equation*}
then it is not difficult to prove that the set
\begin{equation*}
\Gamma =\left\{ \lim_{\lambda \uparrow \Lambda }\ x_{\lambda }\ |\
x_{\lambda }\in \Gamma _{\lambda }\right\}
\end{equation*}
is a hyperfinite grid.

Below $\Gamma$ denotes a hyperfinite grid extending $\mathbb{R}^N$, which \textit{is fixed once and forever}.

\begin{definition}\label{ggg}
The space of grid functions on $\Gamma$ is the family $\mathfrak{G}(\Gamma)$ of functions 
\begin{equation*}
u: \Gamma \rightarrow \mathbb{R}
\end{equation*}
such that, for every $x=\ \lim_{\lambda \uparrow \Lambda }x_{\lambda }\in \Gamma $,
\begin{equation*}
u(x)=\ \lim\limits_{\lambda \uparrow \Lambda }\ u_{\lambda }(x_{\lambda }).
\end{equation*}
\end{definition}

If $f\in \mathfrak{F}(\mathbb{R}^{N})$, and $x=\ \lim_{\lambda \uparrow \Lambda }x_{\lambda }\in \Gamma $, we set 
\begin{equation}
f^{\circ}(x):=\lim\limits_{\lambda \uparrow \Lambda }\ f(x_{\lambda }),
\label{lina}
\end{equation}
namely, $f^{\circ}$ is a restriction to $\Gamma$ of the natural extension $f^{\ast }$ which is defined on the whole $\mathbb{E}^{N}$.

It is easy to check that, for every $a\in \Gamma$, the characteristic function $\chi _{a}(x)$ of $\{a\}$ is a grid function. In perfect analogy with the classical finite case, every grid function can be represented by the following hyperfinite sum,
\begin{equation}
u(x)=\sum_{a\in \Gamma }u(a)\chi _{a}(x)  \label{pu},
\end{equation}
namely, $\left\{ \chi _{a}(x)\right\} _{a\in \Gamma }$ is a set of generators for $\mathfrak{G}(\mathbb{R}^{N})$. 

Moreover, from Definition \ref{ggg} it follows that actually $\left\{ \chi _{a}(x)\right\} _{a\in \Gamma }$ is a basis for the space of grid functions on $\Gamma$.

In general, if $E$ is a subset of $\mathbb{R}^{N}$ and $f$ is defined only on $E$, we set
\begin{equation*}
f^{\circ} (x)=\sum_{a\in \Gamma \cap E^{\ast }}f^{\ast }(a)\chi _{a}(x),
\end{equation*}
namely, we set it to be $0$ on every grid point that does not belong to the natural extension of its domain. For example, if $f(x)=\frac{1}{|x|}$, then $f^{\circ}(0)=0$.

\bigskip

\section{Ultrafunctions\label{Ultrafunctions}}

In this section, we introduce the key ingredient needed for the description of a quantum system, the space of ultrafunctions. An explicit technical construction of (several) ultrafunctions spaces has been provided in Ref.~\cite{ultra, belu2012, belu2013, milano, algebra, beyond, BLS} by various reformulations of nonstandard analysis. However, we prefer to pursue an axiomatic approach to underscore the key properties of ultrafunctions needed for our aims since such technicalities are not important in the applications of the ultrafunctions spaces. For an explicit construction, we refer to \cite{benci}. The basic idea one has to keep in mind is that, in order to build a space of ultrafunctions, we start with a classical space of functions $V(\Omega)$ and extend it to a space of grid functions $V^{\circ}(\Omega)$ with several ad hoc properties.

\subsection{Axiomatic definition of ultrafunctions\label{adu}}

In order to build a space of ultrafunctions that suites for an adequate description of a quantum system, we have to fix an appropriate function space $V(\Omega)$. In that way, we choose $V(\Omega) \supseteq \mathcal{C}^{0}(\Omega)$ to be a function space that includes infinitely differentiable real functions and is a subspace of the space of locally integrable functions (which includes real $p$-integrable functions with $p\geq 1$),
\begin{equation*}
\mathcal{C}^{0}(\Omega)\subset\mathcal{D}(\Omega)\subset V(\Omega) \subset \mathcal{L}_{loc}^{1}(\Omega).
\end{equation*}
In the further step, we build a family of all \textit{\textbf{finite}}-dimensional subspaces of the chosen function space $V(\Omega)$. We label this family by $\left\{ V_{\lambda }\right\} _{\lambda \in \mathfrak{L}}$, so that it has the following property,
\begin{equation*}
\forall\lambda_{1},\lambda_{2}\in\mathfrak{L} \;\;\;\; \lambda _{1}\subseteq \lambda _{2}\Rightarrow V_{\lambda _{1}}\subseteq
V_{\lambda _{2}}.
\end{equation*}
Such a family provides a net $\left\{ V_{\lambda }\right\} _{\lambda \in \mathfrak{L}}$ of the finite subspaces of $V(\Omega)$. Hence it allows us to perform a $\Lambda$-limit of this net resulting in the required hyperfinite function space, the space of ultrafunctions,
\begin{equation}
V^\circ(\Omega) := \lim\limits_{\lambda \uparrow \Lambda} V_\lambda.
\end{equation}
This leads to a clear definition of the space of ultrafunctions, which we equip with the axioms providing some properties needed in quantum mechanics.
\begin{definition} \label{def ultra}
A \textbf{space of ultrafunctions $V^{\circ}(\Omega)$} generated by $V(\Omega)$ and modelled on the family of its finite subspaces $\left\{ V_{\lambda}\right\} _{\lambda \in \mathfrak{L}}$ is a space of grid functions 
\begin{equation*}
u:\Gamma \rightarrow \mathbb{E}
\end{equation*}
equipped with an internal functional
\begin{equation*}
\oint :V^{\circ}(\Omega) \rightarrow \mathbb{E}
\end{equation*}
(called \textbf{pointwise integral}) and $N$ internal operators 
\begin{equation*}
D_{i}:V^{\circ}(\Omega) \rightarrow V^{\circ}(\Omega)
\end{equation*}
(called \textbf{generalized partial derivative}), which satisfy the axioms below.

\begin{axiom}
\label{1}$\forall u\in V^{\circ}(\Omega)$, there exists a net $u_{\lambda }$ such that $\forall \lambda\in\mathfrak{L} \;\; u_{\lambda }\in V_{\lambda }$, and 
\begin{equation*}
u = \lim_{\lambda \uparrow \Lambda } u_{\lambda }.
\end{equation*}
\end{axiom}

\begin{axiom}
\label{2} $\forall u = \lim_{\lambda \uparrow \Lambda }u_{\lambda }$, $u_{\lambda}\in V_{\lambda }$, its pointwise integral is defined as
\begin{equation}
\oint u(x)~dx \; = \; \lim_{\lambda \uparrow \Lambda }\int_{Q_{\lambda}}u_{\lambda }~dx  \label{int},
\end{equation}
where
\begin{equation*}
Q_{\lambda }=\left\{ x\in \mathbb{R}^{N}\ \Bigl|\ \left\vert x_{i}\right\vert \leq |\lambda |\right\}.
\end{equation*}
\end{axiom}

\begin{axiom}
\label{3}$\forall a\in \Gamma $, there exists a positive infinitesimal number $d(a)$ such that
\begin{equation*}
d(a) = \oint \chi _{a}(x)dx>0.
\end{equation*}
\end{axiom}

\begin{axiom}
\label{4}If $x= \lim_{\lambda\uparrow \Lambda }x_{\lambda }\in \Gamma $ and $u = \lim_{\lambda \uparrow \Lambda } u_\lambda$ such that $\forall \lambda\in\mathfrak{L} \;\; u_{\lambda } \in V_{\lambda }\cap C_{c}^{1}(\mathbb{R}^{N})$, then its generalized partial derivative at the point $x$ is defined as
\begin{equation}
D_{i}u(x)=\lim_{\lambda \uparrow \Lambda }\ \partial _{i}u_{\lambda
}(x_{\lambda }).  \label{der}
\end{equation}
\end{axiom}

\begin{axiom}
\label{5}If $D:  V^{\circ}(\Omega) \rightarrow \left( V^{\circ} (\Omega) \right) ^{N}$ is the generalized gradient, i.e.,
\begin{equation*}
Du:=( D_{1} u \; ... \; D_{N} u ),
\end{equation*}
then
\begin{equation*}
Du=0 \; \Leftrightarrow \; u = \operatorname{const}.
\end{equation*}
\end{axiom}

\begin{axiom}
If we set the support of an ultrafunction $u$ as 
\begin{equation*}
\mathfrak{supp}\left( u\right) =\left\{ x\in \Gamma \ |\ u(x)\neq 0\right\} ,
\end{equation*}
then
\begin{equation*}
\mathfrak{supp}\left( D_{i}\chi _{a}(x)\right) \subset \mathfrak{mon}(a).
\end{equation*}
\end{axiom}

\begin{axiom}
\label{7}For every $u,v\in V^{\circ}(\Omega)$, 
\begin{equation}
\oint (D_i u(x)) v(x)dx\; = \; -\oint u(x) (D_i v(x)) dx.  \label{parti}
\end{equation}
\end{axiom}

\end{definition}

Let us notice that, in most previous papers, ultrafunctions spaces were assumed to witness only some of the above axioms. In fact, constructing a space of ultrafunctions that witnesses all these axioms presents several technical challenges. However, these properties are fundamental to develop many applications better, as we are going to show. For more details on the construction of such a space, as well as for the relevance of these axioms, we refer to \cite{benci}.

\subsection{Discussion of the axioms of the space of ultrafunctions\label{aduB}}
\textbf{Axiom 1} means that every ultrafunction is a grid function (based on $V(\Omega)$), which has the following peculiar property: at every step $\lambda$ the corresponding finite-dimensional vector space which contains all $u_{\lambda}$ is $V_{\lambda}$. Hence, every ultrafunction is a $\Lambda$-limit of functions in $V_\lambda$.

\textbf{Axiom 2} is a definition of the pointwise integral. This integral extends the usual Riemann integral from functions in $\mathcal{C}_{c}^{0}(\mathbb{R}^{N})$ to ultrafunctions in $V^{\circ}(\Omega)$. In fact, by definition, $\forall f\in \mathcal{C}_{c}^{0}(\mathbb{R}^{N})$, 
\begin{equation*}
\oint f{^\circ}(x)~dx \; = \; \int f(x)~dx,
\end{equation*}
since for the functions of a compact support the net $\lim_{\lambda \uparrow \Lambda }\int_{Q_{\lambda}}u~dx$ becomes constantly equal to $\int u(x)dx$. Moreover, if $f\in L^{1} (\mathbb{R}^{N})\cap V$, then 
\begin{equation}
\nonumber \oint f(x)~dx \;\sim\; \int f(x)~dx,
\end{equation}
as the quantity 
\begin{equation}
\nonumber \Biggl|\int\limits_{Q_{\lambda}} f(x)~dx \;- \int\limits_{\mathbb{R}^{N}} f(x)~dx\Biggr|
\end{equation}
becomes arbitrarily small for large $\lambda$.

\textbf{Axiom 3} shows that the equality $\oint f^\circ(x)~dx = \int f(x)~dx$ cannot hold for all arbitrary Riemann integrable functions. This is the reason why the new notation $\oint$ has been introduced. A key example (being very important when dealing with singular potentials) is the characteristic function of a singleton. In fact, if $a\in\mathbb{R}^N\cap\Gamma$, then
\begin{eqnarray}
\nonumber \int \chi _{a}(x)dx &=& 0 \\
 &\neq& \oint \chi _{a}(x)dx \; = \;\varepsilon \; \sim \; 0,
\end{eqnarray}
with $\varepsilon>0$, so that it represents the ``weight'' of the point $a$. At the first sight, this Axiom might seem unnatural. However, we work in a non-Archimedean approach, thus, the infinitesimals cannot be forgotten as is done in the Riemann integration. 

\textbf{Axiom 4} shows that the generalized derivative extends the usual derivative. In fact, if $f \in C^{1}(\mathbb{R}^N)$ and $x\in \mathbb{R}^N$, then
\begin{equation*}
D_i f^{\circ}(x) = \lim_{\lambda \uparrow \Lambda}\ \partial_i f(x) = \partial_i f(x).
\end{equation*}

However, the less intuitive fact is that the operator $D_i$ is defined on all ultrafunctions. The last three axioms have been introduced to highlight that all the most useful properties of the usual derivative are also satisfied by $D_{i}$. In fact, \textbf{Axiom 5} says that the ultrafunctions behave as compactly supported $\mathcal{C}^{1}$ functions.

\textbf{Axiom 6} states that the derivative is a local operator\footnote{We do not want to enter too much into details here, as ``locality'' is all we need to develop the applications we have in mind. That said, using some basic tools of nonstandard analysis (underspill and saturation), it would be simple to prove that, actually, there exists an infinitesimal number $\eta$ such that, for every $a\in\Gamma$, the expansion of the derivative $D\chi_{a}(x)$ in the base $\{\chi_{b}\}_{b\in\Gamma}$ would involve only the points $b\in [a-\eta,a+\eta]$. Therefore, the second derivative of $\chi_{a}(x)$ would involve only the points $b\in [a-2\eta,a+2\eta]$ and, more in general, the $n$-th derivative would involve the points in $[a-n\eta,a+n\eta]$. This shows that, for every $n\in\mathbb{N}^{\ast}$ such that $n\eta\sim 0$, the operator $D^{n}$ would still be local. This includes some infinite $n$, but not all infinite $n$, e.g., let $n>\frac{1}{\eta}$. In particular, the infinite matrix $M_{D}$ that corresponds to the operator $D$ in the base $\{\chi_{b}\}_{b\in\Gamma}$ is close to be diagonal, in the sense that if we let $N=\max |\{[a-\eta]\cap\Gamma|,|[a+\eta]\cap\Gamma|\mid a\in\Gamma\}$, then in every row of $M_{D}$ the only non-zero elements are the $M_{n,m}$ with $m\in [n-N,n+N]$. This is similar to the computational approximations of the derivative.}, which is a fundamental fact in the applications of ultrafunctions to quantum mechanics.

\textbf{Axiom 7} provides a weak form of Leibniz rule, which is of primary importance in the theory of weak derivatives, distribution, calculus of variations, etc. We express this rule in its weak form since the Leibniz rule
\begin{equation*}
D(fg) = (Df) g + fDg
\end{equation*}
cannot be satisfied by every ultrafunction because of the Schwartz impossibility theorem (see \cite{Schwartz, algebra}). In fact, if Leibniz rule held for all ultrafunctions, we would construct a differential algebra extending $\mathcal{C}^{0}$, that is not possible. In some sense, ultrafunctions provide a ``solution'' of Schwartz impossibility theorem at the cost of using the weak formulation for Leibniz rule instead of the full one\footnote{The reader interested in a deeper discussion of this fact is referred to \cite{Schwartz}. Let us also mention that Colombeau functions provide an alternative ``solution'' of the Schwartz theorem, see \cite{col85}. In the Colombeau approach, Leibniz rule is preserved. However, we have to take $\mathcal{C}^{\infty}$ functions instead of $\mathcal{C}^{k}$ functions.}. 

\subsection{The structure of the space of ultrafunctions}

Since the space of ultrafunctions is generated by $\{\chi_{a}\}_{a\in\Gamma}$, it follows from \textbf{Axiom 3} that the ultrafunctions integral $\oint$ is actually a hyperfinite sum. In fact, for every $u\in V^{\circ}$, its integral can be expanded as
\begin{equation*}
\oint u(x)~dx = \sum_{a\in \Gamma }u(a)d(a).
\end{equation*}
In order to analyze the applications of the space of ultrafunctions to quantum mechanics we consider the complex-valued ultrafunctions, so that we provide a complexification of the space of ultrafunctions,
\begin{equation*}
\mathcal{H}^{\circ}:=V^{\circ}\oplus iV^{\circ}.
\end{equation*}
Coherently with this notation, we let
\begin{equation*}
\mathcal{H} := V \oplus iV
\end{equation*}
and
\begin{equation*}
\mathcal{H}_\lambda := V_\lambda \oplus iV_\lambda.
\end{equation*}
The pointwise integral allows to define the following sesquilinear form with values in $\mathbb{C}^{\ast}$ on $\mathcal{H}^{\circ}$,
\begin{equation}
\oint u(x)\overline{v(x)}~dx = \sum_{x\in\Gamma} u(x)\overline{v(x)}d(x),
\label{rina}
\end{equation}
where $\overline{z}$ is a complex conjugation of $z$. Due to \textbf{Axiom 3} this form is a scalar product. Moreover, if $u,v,u\cdot v\in V\cap L^{2}\cap\mathcal{C}^{0}$, then 
\begin{eqnarray}
\nonumber \int u(x)\cdot \overline{v(x)}~dx &\sim& \oint u(x)\overline{v(x)}~dx \\
&=&\sum_{x\in \Gamma }u(x)\overline{v(x)}d(x),
\end{eqnarray}
with $\sim$ substituted by equality when $u,v,u\cdot v\in L^{2}\cap\mathcal{C}^{0}_{c}\left(\mathbb{R}^{N}\right)$. This means that this sesquilinear form extends the usual $L^{2}$ scalar product.

The norm of an ultrafunction is given by 
\begin{equation*}
\left\Vert u\right\Vert = \left( \oint |u(x)|^2~dx\right) ^{\frac{1}{2}} = \left(\sum_{a\in\Gamma}|u(a)|^2 d(a)\right).
\end{equation*}

In the theory of ultrafunctions, the Dirac delta function has a simpler interpretation due to the pointwise integral. In fact, we can define the \textbf{delta ultrafunction} (called also the Dirac ultrafunction) as follows. For every $a\in \Gamma$,
\begin{equation*}
\delta_a (x) = \frac{\chi_a (x)}{d(a)}.
\end{equation*}
Our choice can be easily motivated, since, for every $u\in V^{\circ}$,
\begin{eqnarray}
\nonumber \oint \delta_a (x) u(x)~dx &=& \sum_{x\in \Gamma } u(x) \delta_a(x) d(x) \\
\nonumber &=& \sum\limits_{x\in \Gamma } u(x) \frac{\chi_a (x)}{d(a)} d(x) \\
&=& u(a).
\end{eqnarray}

In particular, if $u=f^{\ast}|_{\Gamma}$ for some $f \in \mathcal{D}(\Omega)$, this shows that the scalar product between $\delta_{a}$ and $u$ equals $f^{\ast}(a)$. In particular, if $a\in\mathbb{R}$ then one recovers the classical expected property of a delta function.

Moreover, as in the ultrafunctions framework delta functions are actual functions (and not functionals, like in distributions theory), we can perform on them all the classical operations that do not have sense in the standard analysis like, for example, $\delta^2(x)$.

As $\mathcal{D}\subseteq V$, this shows that the delta ultrafunction behaves as the delta distribution when tested against functions in $\mathcal{D}$. Moreover, delta functions are mutually orthogonal with respect to the scalar product (\ref{rina}). Hence, being normalized they provide an orthonormal basis, called \textbf{delta-basis}, given by
\begin{equation}
\left\{ \sqrt{\delta _{a}}\right\} _{a\in \Gamma }=\left\{ \frac{\chi
_{a}}{\sqrt{d(a)}}\right\} _{a\in \Gamma }. \label{db}
\end{equation}
Hence, every ultrafunction can also be expanded in the following way,
\begin{equation}
u(x)=\sum_{a\in \Gamma }\left( \oint u(\xi )\delta _{a}(\xi )d\xi \right)
\chi _{a}(x).  \label{eq:lella}
\end{equation}

The scalar product allows the following ultrafunctions version of the Riesz representation theorem.
\begin{proposition}
\label{dual}If 
\begin{equation*}
\Phi :\mathcal{H}\rightarrow \mathbb{C}
\end{equation*}
is a linear internal functional, then there exists $u_\Phi$ such that, for all $v = \lim_{\lambda \uparrow \Lambda } v_\lambda \in \mathcal{H}^\circ$,
\begin{equation*}
\oint u_\Phi v~dx = \lim_{\lambda \uparrow \Lambda}\ \Phi( v_\lambda),
\end{equation*}
and, for every $f\in V$,
\begin{equation*}
\oint u_\Phi f^\circ~dx = \Phi(f).
\end{equation*}
\end{proposition}

\begin{proof} If $v\in \mathcal{H}_{\lambda },$ then the map
\begin{equation*}
v\mapsto \Phi(v)
\end{equation*}
is a linear functional over $\mathcal{H}_{\lambda }$ and hence, since there exists $u_{\lambda }\in \mathcal{H}_{\lambda }$ such that $\forall v\in \mathcal{H}_{\lambda },$
\begin{equation*}
\int u_{\lambda }v~dx=\Phi \left( v\right).
\end{equation*}
If we set
\begin{equation*}
u_\Phi = \lim_{\lambda \uparrow \Lambda} u_\lambda,
\end{equation*}
the conclusion follows.
\end{proof}

As already stated, our goal is to apply the theory of ultrafunctions to quantum mechanics. In this way, we are interested in understanding the relationship between ultrafunctions and $L^{2}$-functions. Even though the scalar product in $V^{\circ}$ can be seen as an extension of the $L^{2}$-scalar product, it is still not clear whether $L^{2}$ functions can be embedded into $V^{\circ}$. The basic idea is to use Eq.~(\ref{lina}) to do such an association. However, this does not work since the $L^{2}$-functions are not defined pointwise. For this reason, we introduce the following definition, which uses a weak form of association.

\begin{definition}
\label{L2}Given a function $\psi \in L^{2}(\Omega)$, we denote by $\psi^{\circ}$ the unique ultrafunction such that, for every $v = \lim_{\lambda \uparrow \Lambda} v_\lambda(x)\in \mathcal{H}^\circ$, 
\begin{equation*}
\oint \psi^\circ v~dx = \lim_{\lambda \uparrow \Lambda} \int_{\Omega} \psi v_\lambda~dx.
\end{equation*}
\end{definition}
Proposition \ref{dual} ensures that the above definition is well posed, as the map
\begin{equation*}
\Phi :v\mapsto \int \psi v \; dx
\end{equation*}
is a functional on the space $\mathcal{H}^\circ$.

\subsection{Ultrafunctions and distributions}

Distributions can be easily embedded into ultrafunctions spaces by identifying them with equivalence classes~\cite{BLS}.

\begin{definition}
\label{DEfCorrespondenceDistrUltra}The space of \textbf{\emph{generalized distributions}} on $\Omega$ is defined as follows,
\[
\mathcal{D}^{\prime}(\Omega)=V^{\circ}(\Omega)/N,
\]
where 
\[
N=\left\{ \tau\in V^{\circ}(\Omega)\ |\ \forall\varphi\in\mathcal{D}(\Omega),\ \oint\tau\varphi\ dx\sim 0\right\} .
\]
\end{definition}
The equivalence class of an ultrafunction $u \in V^{\circ}(\Omega)$ is denoted by $\left[u\right]_{\mathcal{D}(\Omega)}$. It contains all ultrafunctions, whose action on the functions in $\mathcal{D}(\Omega)$ differs from the action of $u$ by at most an infinitesimal quantity. An obvious idea is to identify this action with the action of a distribution. However, it is not directly possible since there are ultrafunctions, whose action does not correspond to the action of any distribution. For example, if $u=\tau\delta_{0}$, where $\tau$ is an infinite number, then $\forall \varphi \in \mathcal{D}(\Omega) \;  \oint u\varphi^* \; dx= \tau \varphi(0)$, which is an infinite quantity, whenever $\varphi(0)$ is different from $0$. 

This issue can be overcome by considering the so-called ``bounded'' ultrafunctions.
\begin{definition}
Let $\left[u\right]_{\mathcal{D}(\Omega)}$ be a generalized distribution. We say that $\left[u\right]_{\mathcal{D}(\Omega)}$ is a bounded generalized distribution if $\forall \varphi \in \mathcal{D}(\Omega)$ the integral $\oint u\varphi^{*}~dx$ is finite. We will denote by $\mathcal{D}_B^\prime(\Omega)$ the set of bounded generalized distributions.
\end{definition}
In turn, the spaces of generalized distributions and bounded generalized distributions can be identified by an isomorphism, as shows the following theorem.
\begin{theorem}
\label{bello}There is a linear isomorphism 
\[
\Phi: \mathcal{D}_B^\prime(\Omega) \rightarrow \mathcal{D}^\prime(\Omega)
\]
such that, for every $\left[u\right]_{\mathcal{D}(\Omega)}\in\mathcal{D}_{B}^{\prime}(\Omega)$ and for every $\varphi\in\mathcal{D}(\Omega)$,
\[
\left\langle \Phi\Bigl(\left[u\right]_{\mathcal{D}(\Omega)}\Bigr), \varphi\right\rangle_{\mathcal{D}(\Omega)} = \operatorname{st}\left(\oint u \varphi^{\ast}~dx\right).
\]
\end{theorem}
\begin{proof}
For the proof see, e.g., \cite{algebra}.
\end{proof}
In this way, any equivalence class $\left[u\right]_{\mathcal{D}(\Omega)} \in \mathcal{D}_{B}^{\prime}(\Omega)$ can be substituted by $\Phi \left(\left[u\right]_{\mathcal{D}}\right) \in \mathcal{D}^\prime(\Omega)$, so that we can define its action on the functions in $\mathcal{D}(\Omega)$,
\begin{eqnarray}
\nonumber \left\langle \left[u\right]_{\mathcal{D}(\Omega)},\varphi\right\rangle _{\mathcal{D}(\Omega)}&:=& \Bigl\langle\Phi([u]_{\mathcal{D}(\Omega)}),\varphi\Bigr\rangle_{\mathcal{D}(\Omega)} \\
\nonumber &=& \operatorname{st} \left(\oint u \varphi^{\ast}~dx\right).
\end{eqnarray}
In particular, if $f\in C_{c}^{0}(\Omega)$ and $f^{\ast}\in\left[u\right]_{\mathcal{D}(\Omega)}$, then, for all $\varphi\in\mathcal{D}(\Omega)$,
\begin{eqnarray}
\nonumber \left\langle \left[u\right]_{\mathcal{D}(\Omega)},\varphi\right\rangle _{\mathcal{D}(\Omega)} &=& \operatorname{st} \left(\oint u\,\varphi^{\ast}~dx\right) \\ 
\nonumber &=& \operatorname{st} \left(\oint f^{\ast}\varphi^{\ast}~dx\right) = \int f\varphi~dx.
\end{eqnarray}

\subsection{Self-adjoint operators on ultrafunctions}

Let an operator
\begin{equation*}
L: \mathcal{H}^\circ \rightarrow \mathcal{H}^\circ
\end{equation*}
be an internal linear operator on the space of ultrafunctions, which is a hyperfinite-dimensional space by definition. In this way, $L$ can be represented by a hyperfinite matrix (viewed as an infinite matrix by the standard analysis), because we can build a $\Lambda$-limit
\begin{equation*}
Lu := \lim_{\lambda \uparrow \Lambda }\ L_{\lambda }u_{\lambda },
\end{equation*}
where the operators
\begin{equation*}
L_{\lambda }:\mathcal{H}_{\lambda }\rightarrow \mathcal{H}_{\lambda }
\end{equation*}
can be represented by finite matrices (as every space $V_{\lambda }$ is finite-dimensional). In particular, if $L$ is a self-adjoint operator,
\begin{equation*}
\oint Lu\overline{v}~dx=\oint u\overline{Lv}~dx,
\end{equation*}
then the matrices $L_{\lambda }$ are Hermitian. Hence, $L$ can be represented by a hyperfinite-dimensional Hermitian matrix. Therefore, the spectrum $\sigma (L)$ of $L$ consists of eigenvalues only, more precisely,
\begin{equation*}
\sigma (L)=\left\{ \lim_{\lambda \uparrow \Lambda }\ \mu _{\lambda }\in 
\mathbb{E}\ |\ \forall \lambda ,\mu _{\lambda }\in \sigma (L_{\lambda
})\right\},
\end{equation*}
and its corresponding normalized eigenfunctions (being $\Lambda$-limits of the corresponding eigenfunctions of $L_\lambda$) form an orthonormal basis of $\mathcal{H}^{\circ}$. In that way, the ultrafunctions approach resembles the finite-dimensional vector spaces approach, in the sense that the distinction between self-adjoint operators and Hermitian operators is not needed. We have proven the following

\begin{theorem}\label{herm} Every Hermitian operator $L:\mathcal{H}^{\circ}\rightarrow \mathcal{H}^{\circ}$ is self-adjoint. \end{theorem}

In \cite{benci}, it is possible to find a detailed analysis of the ultrafunctions formalization of the position and the momentum operators. For our applications to the Schr\"{o}dinger equation, let us notice that the Laplacian operator
\begin{equation*}
\Delta :C^{2}\left( \mathbb{R}^{N}\right) \rightarrow C^{0}\left( \mathbb{R}^{N}\right)
\end{equation*}
has the following expression as its ultrafunctions formulation,
\begin{equation*}
D^{2}=\sum_{j=1}^{N}D_{j}^{2}:V^{\circ}\rightarrow V^{\circ}.
\end{equation*}

For applications to quantum mechanics, the following Hamiltonian operator is fundamental,
\begin{equation}
\hat{H}u(x)=-\frac{1}{2}\Delta u(x)+\mathbf{V}(x)u(x), \label{ham}
\end{equation}
where we assume that the mass of the $i$-th particle $m_{i}=1$ and $\hbar=1$.

There is a deep and important difference between the standard and the ultrafunctions approach to the study of $\hat{H}$. In the classical $L^{2}$-theory, a fundamental problem is a choice of an appropriate potential $\mathbf{V}$ such that (\ref{ham}) makes sense. Moreover, it is fundamental to define an appropriate self-adjoint realization of $\hat{H}$. On the other hand, in the theory of ultrafunctions any internal function $\mathbf{V}:\Gamma \rightarrow \mathbb{E}$ provides a self-adjoint operator on $\mathcal{H}^{\circ}$, given by 
\begin{equation}
\hat{H}^{\circ}u(x):=-\frac{1}{2}D^{2}u(x)+\mathbf{V}(x)u(x),  \label{ham+}
\end{equation}
which always has a discrete (in the $\mathbb{R}^{\ast}$ sense) spectrum that consists of eigenvalues only. Of course, this spectrum will be hyperfinite, of cardinality equal to the cardinality of $\Gamma$ (once the multiplicities of the eigenvalues are taken into account). For these reasons, we believe that the ultrafunctions approach allows a much simpler study of ``very singular potentials'' such as
\begin{eqnarray}
\mathbf{V}(x) &=&\tau \delta _{a}(x),\ \ \tau\in \mathbb{E},  \label{bibi} \\
\mathbf{V}(x) &=&\Omega \chi _{E}(x),\ \ \Omega\in \mathbb{E},  \label{bobo}
\end{eqnarray}
where $\tau$ and $\Omega$ might be infinite numbers. The use of non-Archimedean methods is fundamental to give a reasonable model of these potentials, which have an interesting physical meaning. Particularly, the delta function potential~(\ref{bibi}), which represents a very short-ranged interaction, appears in many physical problems. For example, it can serve as a reasonable model for the interactions between atoms and electromagnetic fields (particularly, dipole-dipole interactions) and interatomic potential in a many-body system (particularly, Bose -- Einstein condensate). In this way, in the following, we will take a look at the quantum system with a delta potential within the ultrafunctions framework. Before to start, we will reformulate the basis of quantum mechanics --- its system of axioms --- through the ultrafunctions.

\bigskip

\section{Ultrafunctions and quantum mechanics\label{QM}}

\subsection{Axioms of QM based on ultrafunctions}

In the following table, we provide a set of axioms of quantum mechanics formulated within the ultrafunctions approach and compare it with the standard axioms\footnote{For the sake of simplicity, we focus on the system of axioms based on pure states of the quantum system. However, it can be straightforwardly generalized to the case of mixed states, which describe a statistical mixture of the pure states.}.

\renewcommand{\arraystretch}{2.0}
\setlength{\tabcolsep}{12pt}
\begin{widetext}
\begin{center}
\begin{tabular}{ p{1.75cm} | p{6.5cm}  p{6.5cm} }
   & \textbf{Standard QM} & \textbf{Ultrafunctions QM} \\
  \hline            
  \textbf{Axiom 1} & A physical state of a quantum system is described by a unit vector $|\psi\rangle$ in a complex Hilbert space $\mathcal{H}$. & A physical state of a quantum system is described by a \textbf{unit complex-valued ultrafunction} $\psi \in \mathcal{H}^\circ$. \\
  \hline            
  \textbf{Axiom 2} & A physical observable $A$ is represented by a linear self-adjoint operator $\hat{A}$ acting in $\mathcal{H}$. & A physical observable $A$ is represented by a linear Hermitian operator $\hat{A}$ acting in $\mathcal{H}^\circ$. \\
  \hline            
  \textbf{Axiom 3} & The only possible outcomes of a measurement of an observable $A$ form a set $\{ \mu_j\}$, where $\mu_j \in \sigma(\hat{A})$ are the (generalized) eigenvalues of $\hat{A}$. & The only possible outcomes of a measurement of an observable $A$ form a set $\{ \operatorname{st}(\mu_j)\}$, where $\mu_j \in \sigma(A)$ are the eigenvalues of the operator $\hat{A}$. \\
  \hline            
  \textbf{Axiom 4} & An outcome $\mu_j$ of a measurement of an observable $A$ can be obtained with a probability \begin{equation}\nonumber P_j = |\langle\psi_j |\psi\rangle|^2,\end{equation} where $|\psi_j\rangle$ is the (generalized) eigenstate associated with the observed (generalized) eigenvalue $\mu_j$. After the measurement, the quantum system is left in the state $|\psi_j\rangle$. & An outcome $\operatorname{st}(\mu_j)$ of a measurement of an observable $A$ can be obtained with a probability \begin{equation}\nonumber P_j = |\langle \psi, \psi_j \rangle|^2,\end{equation} where $\psi_j$ is the eigenstate associated with the observed eigenvalue $\mu_j$. After the measurement, the quantum system is left in the state $\psi_j$. \\
  \hline  
\end{tabular}

\begin{tabular}{ p{1.75cm} | p{6.5cm}  p{6.5cm} }
    \textbf{Axiom 5} & The time evolution of the state of the quantum system is described by the Schr\"{o}dinger equation  \begin{equation}\nonumber i\frac{\partial |\psi \rangle }{\partial t}=\hat{H}|\psi\rangle,\end{equation} where $\hat{H}$ is the Hamiltonian of the system. & The time evolution of the state of the quantum system is described by the Schr\"{o}dinger equation  \begin{equation}\nonumber i\frac{\partial \psi }{\partial t}=\hat{H}^\circ\psi,\end{equation} where $\hat{H}^{\circ}$ is the Hamiltonian of the system. \\
    \hline 
    \textbf{Axiom 6} & --- & In a laboratory only the states associated to a \textit{finite} expectation value of the physically relevant quantities can be realized. These states are called \textbf{physical states}, the rest of the states is called \textbf{ideal states}.
\end{tabular}
\end{center}
\end{widetext}

\subsection{Discussion of the axioms}
\subsubsection{Axiom 1: States}
In the standard formalism, a physical system is described by a unit vector in Hilbert space. The ultrafunctions formulation of \textbf{Axiom 1} guarantees that we use the ultrafunctions space $\mathcal{H}^\circ$ (hence, non-Archimedean mathematics) instead of the Hilbert one $\mathcal{H}$. In particular, working within wave functions, the state $|\psi\rangle$ can be represented by a normalized function $\psi \in L^{2}(\Omega ),\ \Omega \subset \mathbb{R}^{N}$. Since $L^2$ can be embedded in $V^\circ$, there exists a canonical embedding due to the Def.~\ref{L2},
\begin{equation*}
{{}^\circ}: \mathcal{H}\rightarrow \mathcal{H}^\circ,
\end{equation*}
given by
\begin{equation*}
\psi \mapsto \psi^\circ.
\end{equation*}
Since $\mathcal{H}^\circ$ is a space much richer than $\mathcal{H}$, the ultrafunctions framework recovers all standard states and provides more possible states, particularly, the ideal states of \textbf{Axiom 6}.

\subsubsection{Axiom 2: Observables}
The standard and ultrafunctions formulations of \textbf{Axiom 2} highlight many similarities as well as some differences. The main difference is the fact that the ultrafunctions formalism needs von Neumann's notion of the self-adjoint operator no more. In the standard formalism, it is not enough for an operator $\hat{A}$ to be Hermitian -- it has to be self-adjoint ($\hat{A} = \hat{A}^{\dagger}$) so that it can represent a physical observable\footnote{For example, if we consider the system of a particle in an infinite potential well, the momentum operator $\hat{P} = -i\frac{d}{dx}$ is Hermitian but not self-adjoint: the domain of $\hat{P}$ is only a subset of the domain of its adjoint. This fact leads to problems with a physical interpretation of the spectrum of an Hermitian but not self-adjoint operator, which can turn out to be the so-called residual spectrum (being not interpretable physically). The reader interested in a deeper discussion of the difference between Hermitian and self-adjoint operators, and its consequences in the standard quantum mechanics is referred to~\cite{Gieres2000}.}.

In the ultrafunctions formalism, the observables of a quantum system can be represented by \textit{internal Hermitian operators}, which are trivially self-adjoint due to the Theorem~\ref{herm}. Hence, for observables, ``Hermitian'' and ``self-adjoint'' become equivalent.

\subsubsection{Axioms 3, 4: Measurement}
In the standard formalism, the possible measurement outcomes of an observable $A$ belong to the spectrum of the corresponding self-adjoint operator $\hat{A}$, which, generally speaking, is decomposed into two sets, discrete spectrum, and continuous spectrum. While the discrete spectrum contains the eigenvalues of $\hat{A}$, the continuous one is a set of the so-called generalized eigenvalues, which correspond to the generalized eigenstates, the eigenstates which do not belong to $\mathcal{H}$. A typical example is given by the position operator $\hat{q}$ on the Hilbert space $L^2(\mathbb{R})$, whose spectrum is purely continuous (the whole real line $\mathbb{R}$), and the corresponding eigenfunctions $\psi_{q}(x) = \delta(x - q)$ do not belong to $L^2(\mathbb{R})$. In that way, the manner the standard formalism treats the discrete spectrum in does not fit for the continuous spectrum leading to numerous misunderstandings: an introduction of some additional constructions such as rigged Hilbert spaces is needed~\cite{Gieres2000}.

In the ultrafunctions formalism, due to the self-adjointness of internal Hermitian operators, it follows that any observable has exactly $\kappa =\dim ^{\ast }(\mathcal{H}^\circ) = |\Gamma|$ eigenvalues (taking into account their multiplicity). In this way, \textit{no more essential distinction between eigenvalues and a continuous spectrum is needed} since a continuous spectrum can be considered as a discrete spectrum containing eigenvalues infinitely close to each other.

For example, if we consider the eigenvalues of the position operator $\mathbf{q}$ of a free particle, then the eigenfunction relative to the eigenvalue $q\in \mathbb{R}$ is the Dirac ultrafunction $\delta _{q}$. Therefore, it can be trivially seen that its spectrum is $\Gamma$, which is not a continuous spectrum in $\mathbb{R}^*$. However, the standard continuous spectrum can be recovered since the eigenvalues of an internal Hermitian operator $\hat{A}$ are Euclidean numbers. Hence, assuming that measurement gives a real number, we have imposed in the \textbf{Axiom 3} that its outcome is $\operatorname{st}(\mu)$. In the case of the spectrum of the position operator, we can show that
\begin{equation*} 
\{\operatorname{st}(\mu)\mid \mu\in\sigma(\hat{q})\} = \mathbb{R},
\end{equation*}
because every real number lies in $\Gamma$.

Working in a non-Archimedean framework, we have set in a natural way that the transition probabilities should be non-Archimedean. In fact, we assume that the probability is better described by the Euclidean number $\left\vert \langle \psi ,\psi _{j} \rangle \right\vert ^{2}$, rather than the real number $\operatorname{st}(\left\vert \langle \psi ,\psi _{j} \rangle \right\vert ^{2})$. For example, let $\psi \in \mathcal{H}^\circ$ be the state of a system. The probability of an observation of the particle in a position $q$ is given by
\begin{equation*}
\left\vert \doint \psi (x)\delta _{q}(x)\sqrt{d(q)}~dx\right\vert^2  = \left\vert \psi (q)\right\vert^2 d(q),
\end{equation*}
which is an infinitesimal number. The standard probability can be recovered by the means of the standard part, in the case of the position operator it would be zero (as is expected). We refer to \cite{BHW,BHW16} for a presentation and discussion of the non-Archimedean probability.

\subsubsection{Axiom 5: Evolution}

The ultrafunctions version of this axiom is very similar to the standard one. Since $\hat{H}^\circ$ is an internal operator defined on a hyperfinite-dimensional vector space $\mathcal{H}^\circ$, it can be represented by \textit{an Hermitian hyperfinite matrix} due to the Theorem~\ref{herm}. Hence, the evolution operator of the quantum system is described by the exponential matrix $\hat{U}^\circ (t) = e^{-i\hat{H}^\circ t}$.

\subsubsection{Axiom 6: Physical and ideal states}

This is the most peculiar axiom in the ultrafunctions approach. In the ultrafunctions theory, the mathematical distinction between the physical eigenstates and the ideal eigenstates is intrinsic. It does not correspond to anything in the standard formalism. Hence, it opens a very interesting problem of the physical relevance of such states. Basically, we can intuitively say that the physical states correspond to the states which can be prepared and measured in a laboratory, while the ideal states represent ``extreme'' states useful in the foundations of quantum mechanics, thought experiments (Gedankenexperimente) and computations.

For example, a Dirac ultrafunction is not a physical state but an ideal state. It represents a state in which the position of the particle is perfectly determined, which is precisely what one has in mind considering the Dirac delta distribution. Clearly, this state cannot be produced in a laboratory since it requires infinite energy. However, it is useful in our description of the physical world since such a state makes more explicit the standard approach. Therefore, \textbf{Axiom 6} highlights a concept that is already present (but somehow hidden) in the standard approach.

For example, in the Schr\"{o}dinger representation of a free particle in $\mathbb{R}^{3}$, consider the state 
\begin{equation*}
\psi (x)=\frac{\varphi(x)}{|x|},\ \varphi \in \mathcal{D}(\mathbb{R}^{3}),\ \varphi(0) > 0.
\end{equation*}
We see that $\psi (x)\in L^{2}(\mathbb{R}^{3})$ but this state cannot be produced in a laboratory, since the expected value of its energy
\begin{equation*}
\langle \hat{H}\psi, \psi \rangle = \frac{1}{2}\int \left\vert \nabla \psi \right\vert ^{2}~dx
\end{equation*}
is infinite (even if the result of a single experiment is a finite number). 

\section{Example: Hamiltonian with a $\delta$-potential\label{deltaP}}

To compare the standard and the ultrafunctions approaches, in this section, we want to study the Schr\"{o}dinger equation
\begin{equation}\label{thisone}
\hat{H}\psi \equiv -\frac{1}{2}\Delta \psi + \tau\delta_0(x)\psi = E\psi,
\end{equation}
which corresponds to the problem of a particle moving in the Dirac delta potential of transparency $\tau$ at the point $x=0$. As we are going to show, there are two main differences between the standard and the ultrafunctions approaches,
\begin{enumerate}[label={\arabic*)}]
\item the standard approach changes if we change the dimension of the space (as in the space of dimension $\mathscr{D}>1$ one of the main problems is to find a self-adjoint representation of $\hat{H}$), while in the ultrafunctions approach we always have a self-adjoint representation of $\hat{H}$, independently of $\mathscr{D}$, due to Theorem \ref{herm},
\item in the ultrafunctions approach the constant $\tau$ can be an infinitesimal or an infinite number, which allows us to construct the models which cannot be considered within the standard framework (for example, a potential equal to the characteristic function of a point\footnote{In fact, within standard approach such a potential would be indistinguishable from 0.}).
\end{enumerate}

\subsection{The standard approach}

Let us review the solution of the problem within the standard approach (namely, $\tau\in\mathbb{R}$)~\cite{Fluegge,Belloni}. At first, it is assumed that the particle moves in a box of length $2L$. In turn, the delta potential is approximated as a square wall (well) of a finite length $2\varepsilon$ and a finite height $V_0$,
\begin{eqnarray}
 V_\varepsilon(x) &=& V_0 \theta(|\varepsilon-x|) \xlongrightarrow{V_0 \rightarrow \infty, \; \varepsilon \rightarrow 0} \tau \delta_0(x), \label{Approx1D}
\end{eqnarray}
where the transparency of the potential is related to the parameters of the approximation as $\tau = 2\varepsilon V_0$. In that way, within the standard approach, one solves the Schr\"{o}dinger equation with the approximated potential,
\begin{equation}\label{thisone_approx}
-\frac{1}{2}\Delta \psi + V_\varepsilon(x)\psi = E\psi,
\end{equation}
and, in the obtained solution, takes the limit $\varepsilon\to 0$ and expands the box by the limit $L \to \infty$, in order to recover the original problem.

\subsubsection{Potential barrier ($\tau>0$)}

We start with the case of a potential barrier with $\tau>0$, which is illustrated on the Fig.~\ref{Fig1}.
\begin{center}
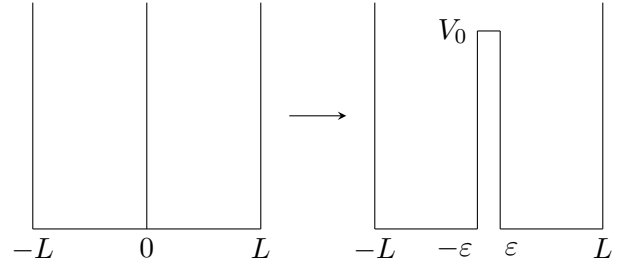
\begin{figure}[h!]
\begin{tikzpicture}[every node/.style={outer sep=0pt}, scale=0.75]
\draw (-2,0) --(-2,4);
\draw (0,0) --(0,4);
\draw (2,0) --(2,4);
\draw (-2,0) --(2,0);
\node [below] at (-2,0) {$-L$};
\node [below] at (0,0) {$0$};
\node [below] at (2,0) {$L$};
\draw [-stealth] (2.5,2) -- (3.5,2);
\draw (4,0) --(4,4);
\draw (5.8,0) --(5.8,3.5);
\draw (6.2,0) --(6.2,3.5);
\draw (8,0) --(8,4);
\draw (4,0) --(5.8,0);
\draw (6.2,0) --(8,0);
\draw (5.8,3.5) --(6.2,3.5);
\node [below] at (4,0) {$-L$};
\node [below] at (5.4,0) {$-\varepsilon$};
\node [below] at (6.4,-0.075) {$\varepsilon$};
\node [below] at (8,0) {$L$};
\node [left] at (5.8,3.5) {$V_0$};
\end{tikzpicture}
\caption{Delta potential as a limit of a finite barrier.}
\label{Fig1}
\end{figure}
\end{center}
Using the abbreviations $k^2 = 2E$ and $\varkappa^2 = 2(V_0 - E)$ we obtain the following solution of the Schr\"{o}dinger equation,
\begin{eqnarray}
    \nonumber \psi^{\pm}(x)&=&\left\{
                \begin{array}{ll}
                  \mp A^{\pm} \sin(k(x+L)), \; x\in[-L, -\varepsilon],\\
                  B^{\pm} f^\pm(x), \; x\in[-\varepsilon, \varepsilon], \\
                  A^{\pm} \sin(k(x-L)), \; x\in[\varepsilon, L],
                \end{array}
              \right.
\end{eqnarray}
where $A^{\pm}$ and $B^{\pm}$ are the normalization constants, $f^+(x) = \cosh(\varkappa x)$, $f^-(x) = \sinh(\varkappa x)$, ``$+$'' corresponds to the even function and ``$-$'' corresponds to the odd function. We seek the energies of the eigenstates, which can be found with use of the continuity condition,
\begin{equation}
 \frac{\psi'(x)}{\psi(x)} \Biggr|_{x \rightarrow \varepsilon - 0} = \frac{\psi'(x)}{\psi(x)} \Biggr|_{x \rightarrow \varepsilon + 0}. \label{Continuity}
\end{equation}

In turn, the original problem is recovered by taking the limit~(\ref{Approx1D}), which leads to the following solution,
\begin{eqnarray}
    \nonumber \psi^{\pm}_n(x)&=&\left\{
                \begin{array}{ll}
                  \mp A^{\pm}_n \sin(k^{\pm}_n(x+L)), \; x\in[-L, 0],\\
                  A^{\pm}_n \sin(k^{\pm}_n(x-L)), \; x\in[0, L],
                \end{array}
              \right.
\end{eqnarray}
with the normalization $|A^{\pm}_n|^{-2} = \frac{2k^{\pm}_n L - \sin(2k^{\pm}_n L)}{2k^{\pm}_n}$ and the quantities $k_n^{\pm}$, which are defined by the following equations due to Eq.~(\ref{Continuity}),
\begin{eqnarray}
k^{+}_n \cot(k^{+}_n L) &=& -\tau, \\
k^{-}_n \cot(k^{-}_n L) &=& -\infty \;\;\Rightarrow\;\; k^{-}_n = \frac{\pi + 2\pi n}{L},
\end{eqnarray}
and, in turn, define the energies $E_n = \frac{k_n^2}{2}$ of the eigenstates. At the singularity point, we obtain
\begin{eqnarray}
    \psi_n^{+}(0)&=&-A^{+}_n \sin(k^{+}_n L), \\
    \psi_n^{-}(0)&=&0.
\end{eqnarray}

\subsubsection{Potential well ($\tau<0$)}

Let us now consider the case of a potential well (illustrated on the Fig.~\ref{Fig2}), namely, the negative transparency $\tau<0$. In this case, we have two different situations with respect to energy, $E > 0$ (scattering states) and $E < 0$ (bound states). 

\begin{center}
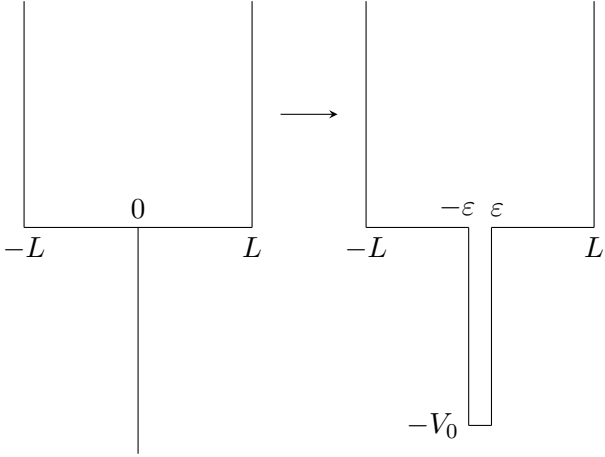
\begin{figure}[h!]
\begin{tikzpicture}[every node/.style={outer sep=0pt}, scale=0.75]
\draw (-2,0) --(-2,4);
\draw (0,0) --(0,-4);
\draw (2,0) --(2,4);
\draw (-2,0) --(2,0);
\node [below] at (-2,0) {$-L$};
\node [above] at (0,0) {$0$};
\node [below] at (2,0) {$L$};
\draw [-stealth] (2.5,2) -- (3.5,2);
\draw (4,0) --(4,4);
\draw (5.8,0) --(5.8,-3.5);
\draw (6.2,0) --(6.2,-3.5);
\draw (8,0) --(8,4);
\draw (4,0) --(5.8,0);
\draw (6.2,0) --(8,0);
\draw (5.8,-3.5) --(6.2,-3.5);
\node [below] at (4,0) {$-L$};
\node [above] at (5.6,0) {$-\varepsilon$};
\node [above] at (6.3,0.0375) {$\varepsilon$};
\node [below] at (8,0) {$L$};
\node [left] at (5.8,-3.5) {$-V_0$};
\end{tikzpicture}
\caption{Delta potential as a limit of a finite well.}
\label{Fig2}
\end{figure}
\end{center}

\paragraph{Positive energies ($E > 0$): scattering states.}
Using the abbreviations $k^2 = 2E$ and $\varkappa^2 = 2(V_0 + E)$ we obtain the following solution of the Schr\"{o}dinger equation,
\begin{eqnarray}
    \nonumber \psi^{\pm}(x)&=&\left\{
                \begin{array}{ll}
                  \mp A^{\pm} \sin(k(x+L)), \; x\in[-L, -\varepsilon],\\
                  B^{\pm} f^{\pm}(x), \; x\in[-\varepsilon, \varepsilon], \\
                  A^{\pm} \sin(k(x-L)), \; x\in[\varepsilon, L],
                \end{array}
              \right.
\end{eqnarray}
where $A^{\pm}$ and $B^{\pm}$ are the normalization constants, $f^+(x) = \cos(\varkappa x)$, $f^-(x) = \sin(\varkappa x)$, ``$+$'' corresponds to the even function and ``$-$'' corresponds to the odd function. By performing the limits as done above one obtains the following solution,
\begin{eqnarray}
    \nonumber \psi^{\pm}_n(x)&=&\left\{
                \begin{array}{ll}
                  \mp A^{\pm}_n \sin(k^{\pm}_n(x+L)), \; x\in[-L, 0],\\
                  A^{\pm}_n \sin(k^{\pm}_n(x-L)), \; x\in[0, L],
                \end{array}
              \right.
\end{eqnarray}
with the normalization $|A^{\pm}_n|^{-2} = \frac{2k^{\pm}_n L - \sin(2k^{\pm}_n L)}{2k^{\pm}_n}$ and the quantities $k_n^{\pm}$, which are defined by the equations
\begin{eqnarray}
k^{+}_n \cot(k^{+}_n L) &=& \tau, \\
k^{-}_n \cot(k^{-}_n L) &=& \infty \;\;\Rightarrow\;\; k^{-}_n = \frac{2\pi n}{L},
\end{eqnarray}
and, in turn, define the energies $E_n = \frac{k_n^2}{2}$ of the eigenstates. At the singularity point, we obtain
\begin{eqnarray}
    \psi_n^{+}(0)&=&-A^{+}_n \sin(k^{+}_n L), \\
    \psi_n^{-}(0)&=&0.
\end{eqnarray}

\paragraph{Negative energies ($E < 0$): bound states.}
Using the abbreviations $k^2 = 2|E|$ and $\varkappa^2 = 2(V_0 - |E|)$ we obtain the following solution of the Schr\"{o}dinger equation,
\begin{eqnarray}
    \nonumber \psi^{\pm}(x)&=&\left\{
                \begin{array}{ll}
                  \mp A^{\pm} \sinh(k(x+L)), \; x\in[-L, -\varepsilon],\\
                  B^{\pm} f^\pm (x), \; x\in[-\varepsilon, \varepsilon], \\
                  A^{\pm} \sinh(k(x-L)), \; x\in[\varepsilon, L],
                \end{array}
              \right.
\end{eqnarray}
where $A^{\pm}$ and $B^{\pm}$ are the normalization constants, $f^+(x) = \cos(\varkappa x)$, $f^-(x) = \sin(\varkappa x)$, ``$+$'' corresponds to the even function and ``$-$'' corresponds to the odd function. By performing the limits as done above one obtains the following solution,
\begin{eqnarray}
    \nonumber \psi^{\pm}_n(x)&=&\left\{
                \begin{array}{ll}
                  \mp A^{\pm}_n \sinh(k^{\pm}_n(x+L)), \; x\in[-L, 0],\\
                  A^{\pm}_n \sinh(k^{\pm}_n(x-L)), \; x\in[0, L],
                \end{array}
              \right.
\end{eqnarray}
with the normalization $|A^{\pm}_n|^{-2} = \frac{2k^{\pm}_n L - \sinh(2k^{\pm}_n L)}{2k^{\pm}_n}$ and the quantities $k_n^{\pm}$, which are defined by the equations
\begin{eqnarray}
k^{+}_n \coth(k^{+}_n L) &=& \tau, \label{BoundStates} \\
k^{-}_n \coth(k^{-}_n L) &=& \infty,
\end{eqnarray}
and, in turn, define the energies $E_n = -\frac{k_n^2}{2}$ of the eigenstates. In this way, we find that, for $E<0$, there exists a \textit{unique} value $k^+$ defined by Eq.~(\ref{BoundStates}). With expanding the box by taking the limit $L \rightarrow \infty$, it corresponds to the energy of the \textit{unique} bound state
\begin{equation}
E_{1D} = -\frac{\tau^2}{2}.
\end{equation}
At the singularity point, we obtain
\begin{eqnarray}
    \psi^{+}(0) &=& -A^{+} \sin(k^{+} L).
\end{eqnarray}

\subsubsection{Multidimensional $\delta$-potentials}

In the analysis of the delta potential, we have focused on the 1-dimensional case. In this case, one can easily prove that there exists a self-adjoint realization of the Hamiltonian $\hat{H}$ bounded below. As shown above, there exists a unique bound state for a Hamiltonian with a one-dimensional delta potential well.

The delta potentials of a higher dimensionality $\mathscr{D}$ provide not only a pedagogical model but find their applications in nuclear and condensed matter physics~\cite{Mitra, deLlano, GeltmanDelta, Farrell}. However, solving a Schr\"{o}dinger equation with these potentials is a more cumbersome problem since there is no rigorous construction of a self-adjoint realization of the Hamiltonian in Eq.~(\ref{thisone}) as long as it is not extended to a non-interacting Hamiltonian on a space with a removed point~\cite{AlbeverioBook, Jackiw, Albeverio1995, Albeverio2000, DellAntonio, Scandone}. Let us illustrate it by the case of a two-dimensional and three-dimensional delta potential well.

As in the one-dimensional case, the delta potential well of dimensionality $\mathscr{D}$ can be approximated by a finite square well described by the Heaviside function $\theta(r)$,
\begin{equation}
 V_\varepsilon(r) = -V_0 \theta(\varepsilon-r) \xlongrightarrow{V_0 \rightarrow \infty, \; \varepsilon \rightarrow 0} -\tau \delta_0(r),
\end{equation}
where its depth $V_0$ and length $\varepsilon$ are related to the transparency of the delta potential well as $\tau_{2D} = \pi \varepsilon^2 V_0$ in the two-dimensional case and $\tau_{3D} = \frac{4\pi}{3} \varepsilon^3 V_0$ in the three-dimensional case, respectively~\cite{deLlano, GeltmanDelta}.

The most interesting part of the problem with a delta potential well is the estimation of the bound states. After solving the corresponding Schr\"{o}dinger equation, the bound states can be found from the continuity equation. In the two-dimensional case, it is more difficult to do, because the wave functions are Bessel functions in this case. However, since the Bessel function $J_0(r)$ with an azimuthal symmetry asymptotically behaves as a cosine, the number of the bound states can be approximately estimated as $N = \varepsilon\kappa_{max}/\pi$, where $\kappa^2_{max} = 2V_0$. Plugging in the transparency, we find
\begin{equation}
 N_{2D} \simeq \frac{1}{\pi}\sqrt{\frac{2\tau}{\pi}},
\end{equation}
which is a finite number. The solution for the three-dimensional potential well can be found in the form of trigonometric functions, that simplifies the calculations. The number of bound states is given by~\cite{deLlano, GeltmanDelta}
\begin{equation}
 N_{3D} \simeq \frac{1}{\pi}\sqrt{\frac{3\tau}{2\pi \varepsilon}},
\end{equation}
which is already an infinite number.

The energy of a bound state for the delta potential can be found by solving the following equation~\cite{Cavalcanti, Nyeo},
\begin{equation}
 \frac{1}{\tau} + \frac{1}{(2\pi)^\mathscr{D}}\int \frac{d^\mathscr{D} k}{k^2 - E} = 0,
\end{equation}
where $\mathscr{D}$ is the number of dimensions. The integral in this equation is divergent for $\mathscr{D}\geq2$, that results in infinite energies of the bound states. In particular, in the two-dimensional case, it can be shown that~\cite{deLlano}
\[ E_{2D} \simeq -\frac{1}{2\varepsilon^2}e^{-\frac{2\pi}{\tau}},\] 
which is infinite. In this way, we conclude that multidimensional delta potential wells provide, generally speaking, an infinite number of bound states with the energies diverging to $-\infty$.

In order to avoid the problems with the infinite quantities (diverging energies) within the standard approach, one usually applies the so-called regularization and renormalization procedures to the calculations. The regularization provides a cut-off\footnote{In the renormalization theory, as well as in the references, the cut-off is denoted by $\Lambda$. However, to avoid confusion with the $\Lambda$-limit we prefer to change the notation here, as we will not discuss renormalization in detail in this paper.}$\sigma$ for the divergent integral, and the renormalization re-defines the transparency $\tau$ via $\tau_{R}$ (where $R$ stands for ``renormalized''), which absorbs the dependence on the cut-off in order to absorb the divergence of the integral.

In particular, in the two-dimensional case, the renormalized transparency is defined by the following equation~\cite{Cavalcanti, Nyeo},
\begin{equation}
\frac{1}{\tau_R} = \frac{1}{\tau} + \frac{1}{4\pi} \operatorname{ln}\left(\frac{\sigma^2}{\omega^2}\right),
\end{equation}
where $\omega$ is an arbitrary parameter, which represents the renormalization scale. Then, one takes a limit $\sigma\to\infty$ and varies the ``bare'' transparency $\tau$ in such a way that the renormalized transparency $\tau_R$ remains finite~\cite{Cavalcanti, Nyeo}. Such a renormalized transparency leads to a unique bound state with the binding energy
\begin{equation}
E_{2D} = -\omega^2 e^{-\frac{4\pi}{\tau_R}}.
\end{equation}
Finally, this procedure means that one introduces an additional scale in order to ``forget'' about the unboundedness of the Hamiltonian from below.

\subsection{The ultrafunctions approach}

\subsubsection{Solution in the singularity point}
In the ultrafunctions approach, the Schr\"odinger equation reads
\begin{equation}\label{ultraschro} -D^{2}u+\tau\delta_{0} u=Eu, \end{equation}
where $\tau, E \in \mathbb{R}^*$. We use the linearity of $V^{\circ}$ to write the solution $u$ of Eq.~(\ref{ultraschro}) as 
\[ u(x) = \sum_{a\in\Gamma} u(a) \chi_a(x).\]
Therefore, we can rewrite Eq.~(\ref{ultraschro}) as
\begin{eqnarray}\label{schr1}
&-& \nonumber \sum_{a\in\Gamma} u(a) D^2\chi_a(x) + \tau \sum_{a\in\Gamma} u(a) \chi_a(x) \delta_{0}(x)\\
&=& E \sum_{a\in\Gamma} u(a) \chi_a(x).
\end{eqnarray}
For every $a\in\Gamma$, let $D^{2}\chi_{a}(x):=\sum_{b\in\Gamma}D_{a,b}\chi_{b}(x)$. Notice that the locality of the ultrafunctions derivative entails that, actually, we can write\footnote{Readers familiar with nonstandard analysis will notice that our notation $\sum_{b\in\mathfrak{mon}(a)}$ is not proper, as $\mathfrak{mon}(a)\cap\Gamma$ is not an internal set (hence, in particular, it is not hyperfinite). However, by underspill it is trivial to prove that $D_{a,b}\neq 0$ only for a hyperfinite amount of points $b\in\mathfrak{mon}(a)\cap\Gamma$, as we have already noticed in Section \ref{aduB}, hence ours is just a small abuse of notations.} \[
D^2 \chi_a(x) = \sum_{b \in \Gamma \cap \mathfrak{mon}(a)} D_{a,b} \chi_b(x).
\]
Therefore, since $\sum_{a\in\Gamma} u(a) \delta(a) \chi_a(x) = u(0)\frac{\chi_0(x)}{d(0)}$ we can rewrite Eq.~(\ref{schr1}) as
\begin{eqnarray}
&-& \nonumber \sum_{a\in\Gamma} \left(u(a) \sum_{b\in\Gamma\cap\mathfrak{mon}(a)}D_{a,b}\chi_{b}(x)\right) + \tau u(0)\frac{\chi_{0}(x)}{d(0)} \\
&=& E \sum_{a\in\Gamma} u(a) \chi_a(x).
\end{eqnarray}

Let us explicitly note that the above discussion does not depend on the dimensionality of the space. In fact, for any number of dimensions, Eq.~(\ref{schr1}) reveals the self-adjoint representation of the Hamiltonian in Eq.~(\ref{thisone}) in the ultrafunctions setting.

Outside the monad of $0$, the equation actually reads
\begin{equation}
-\sum_{a\in \mathfrak{mon}(x) \cap \Gamma} u(a) D^2_a(x) = Eu(x),
\end{equation}
which is formally the same that one gets for the equation 
\begin{equation}
\label{ultraschronodelta}
-D^{2}u = Eu.
\end{equation}
On the other hand, at the point $0\in\Gamma$, we obtain the equation
\begin{equation}
\nonumber -\sum_{a\in \mathfrak{mon}(0) \cap \Gamma} u(a) D^2_a(0) + \frac{\tau}{d(0)} u(0) = E u(0),
\end{equation}
which corresponds to 
\begin{equation}
\nonumber u(0) = \frac{\sum\limits_{a\in \mathfrak{mon}(0) \cap \Gamma} \left(u(a) \sum\limits_{b\in\Gamma\cap\mathfrak{mon}(0)}D_{a,b}\chi_{b}(x)\right)}{E-\frac{\tau}{d(0)}}.
\end{equation}

Notice that, for every other point $\gamma \in \mathfrak{mon}(0) \cap \Gamma$, we obtain
\begin{equation}
\nonumber u(\gamma) = \frac{\sum\limits_{a\in \mathfrak{mon}(0) \cap \Gamma} \left(u(a) \sum\limits_{b\in\Gamma\cap\mathfrak{mon}(0)}D_{a,b}\chi_{b}(x)\right)}{E}.
\end{equation}

In particular, it shows that, in the basis $\{\chi_{a}\}_{a\in\Gamma}$, the matrix representation $M_{\delta}$ of Eq.~(\ref{ultraschro}) and the matrix representation $M_{\Delta}$ of Eq.~(\ref{ultraschronodelta}) satisfy the relation $M_{\delta}=M_{\Delta}+H$, where $H$ is a matrix with 
\begin{equation*} H_{a,b}=
\begin{cases}
\frac{\tau}{d(0)}, & \text{if}\ a=b=0,\\
0, & \text{otherwise.}
\end{cases}
\end{equation*}
Particularly, in the case of a one-dimensional potential well with an infinitesimal transparency, there will still exist a unique bound state with infinitesimal negative energy.

\subsubsection{Connection with the standard solutions}
As usual, in order to introduce a new solution concept for a standard problem, it is important to discuss the relationship between the standard and new solution. Our first result shows that the approximation procedure used in most standard approaches to the Schr\"{o}dinger equation with a delta potential can be reproduced within the ultrafunctions approach, in the following sense.
\begin{theorem}
\label{ultraExistence} Let $0\sim a_{\Lambda}=\lim_{\lambda \uparrow \Lambda} a_\lambda$ and $\tau_{\Lambda} = \lim_{\lambda \uparrow \Lambda} \tau_\lambda$. Then there exists a space of ultrafunctions $V^{\circ}$ that contains a solution $u_\Lambda$ of the equation
\begin{equation}
-D^{2}u_\Lambda + \tau_\Lambda\delta_{0} u_\Lambda = E_\Lambda u_\Lambda,
\end{equation}
which is a $\Lambda$-limit of a net of solutions $u_\lambda \in V(\Omega)$ of the standard Schr\"{o}dinger equations 
\begin{equation}\label{approx schr}
-\Delta u_\lambda + \tau_\lambda A_{a_{\lambda}} u_\lambda = E_\lambda u_\lambda,
\end{equation}
where $A_{a_{\lambda}}$ is an approximation of the delta distribution by a square well of length $a_\lambda \in \mathbb{R}$, and, for every $\lambda$, $E_{\lambda}$ is an eigenvalue of the energy of the solution of the standard Schr\"{o}dinger equation with the approximated delta potential such that $E_\Lambda = \lim_{\lambda \uparrow \Lambda} E_\lambda$.
\end{theorem}
\begin{proof} For every $\lambda\in\Lambda$, let $u_{\lambda}$ be a solution of Eq.~(\ref{approx schr}). Let $\{V_{\lambda}\}_{\lambda}$ be an increasing net of finite dimensional subspaces of $V(\Omega)$ as in Definition \ref{def ultra}.
If $\forall \lambda$ the solution $u_\lambda \in V_\lambda$, then a $\Lambda$-limit $u_\Lambda = \lim_{\lambda \uparrow \Lambda} u_\lambda$ is contained in the space of ultrafunctions $V^{\circ}$ generated by the net $\{V_\lambda\}_{\lambda}$, and theorem is proven.

If $u_\lambda \notin V_\lambda$, then we extend $V_\lambda$ by 
\[  V^{\prime}_\lambda := \operatorname{span}(V_\lambda \cup \{ u_\lambda \}), \]
and build the $\Lambda$-limit with the net $\{V^{\prime}\}_\lambda$. In this way, the $\Lambda$-limit of the net $\{u_\lambda\}$ gives an ultrafunction in $V^{\circ}$ by construction.
\end{proof}

The second result we want to show precises the relationship between the ultrafunction and standard solution in the case of finite energy and transparency.
\begin{theorem}
\label{StandardConnection}
Let $u\in V^{\circ}$ be an ultrafunction solution of the Eq.~(\ref{ultraschro}), where the transparency $\tau$ and the energy $E$ are finite numbers in $\mathbb{R}^{\ast}$. Let $w\in\mathcal{D}^{\prime}(\Omega)$ be such that\footnote{In condition 2, we have a distributional product $\delta\cdot w$, which is not defined, in general. We tacitly assume that $w$ is such that this product makes sense as a distribution. This is the case, e.g., when $w$ is smooth. Notice that, in contrast, the product of ultrafunctions is always well defined.}
\begin{enumerate}
 \item $[u]_{\mathcal{D}(\Omega)} = w$,
 \item $[\delta \cdot u]_{\mathcal{D}(\Omega)} = \delta \cdot w$. 
\end{enumerate}
Then $w$ is a solution of the standard Schr\"{o}dinger equation
\begin{equation}-\Delta w + 2\operatorname{st}(\tau) \delta_0 w = \operatorname{st}(E) w. \end{equation}
\end{theorem}
\begin{proof}
By the definition of $[u]_{\mathcal{D}(\Omega)}$, we can write
\begin{equation}
 u = w^\circ + \psi,
\end{equation}
where $\psi$ is such that for any $f \in \mathcal{D}(\Omega)$,
\begin{equation}
 \langle \psi, f^\circ \rangle \sim 0.
\end{equation}
If $f\in\mathcal{D}(\Omega)$, we have
\begin{equation}
 \nonumber \langle w, f \rangle = \langle w^\circ, f^\circ \rangle = \langle u, f^\circ \rangle - \langle \psi, f^\circ \rangle \sim \langle u, f^\circ \rangle.
\end{equation}
Using Eq.~(\ref{ultraschro}) we can write
\begin{equation}
 \nonumber \langle u, f^\circ \rangle = \frac{1}{E} \Bigl( \langle D^2 u , f^\circ \rangle + \tau \langle \delta_0 u, f^\circ \rangle \Bigr).
\end{equation}
The first term can be rewritten in the following way,
\begin{eqnarray}
 \nonumber \frac{1}{E} \langle D^2 u , f^\circ \rangle &=& \frac{1}{E} \langle u , D^2 f^\circ \rangle \\
 \nonumber &\sim& \frac{1}{\operatorname{st}(E)} \langle w , \Delta f \rangle = \frac{1}{\operatorname{st}(E)} \langle \Delta w , f \rangle.
\end{eqnarray}
Using the hypothesis on $[\delta\cdot u]_{\mathcal{D}(\Omega)}$, the second term can be rewritten in the following way,
\begin{eqnarray}
 \nonumber \frac{\tau}{E} \langle \delta_0 u , f^\circ \rangle &\sim& \frac{2\operatorname{st}(\tau)}{\operatorname{st}(E)} \langle \delta_0 u, f^\circ \rangle \sim \frac{2\operatorname{st}(\tau)}{\operatorname{st}(E)} \langle \delta_0 w, f \rangle.
\end{eqnarray}
This means that
\begin{equation}
 \nonumber \langle w, f \rangle = \frac{1}{\operatorname{st}(E)} \langle \Delta w, f \rangle + \frac{2\operatorname{st}(\tau)}{\operatorname{st}(E)} \langle \delta_0 w, f \rangle,
\end{equation}
or, equivalently,
\begin{equation}
 \nonumber \Delta w + 2\operatorname{st}(\tau) \delta_0 w = \operatorname{st}(E) w,
\end{equation}
which is the corresponding standard Schr\"{o}dinger equation.
\end{proof}

\subsubsection{Bounds for the energy}
Some easy bounds on the energy of the solution of Eq.~(\ref{ultraschro}) can be found using the scalar product of ultrafunctions. By Axiom 7 of the space of ultrafunctions, 
\begin{equation}
\langle \Delta u(x), u(x) \rangle = - \langle Du(x), Du(x) \rangle \leq 0.
\end{equation}
Using the Schr\"{o}dinger equation we can rewrite the previous equation, in order to get
\begin{equation}
\nonumber \langle \Delta u(x), u(x) \rangle = -2\Bigl\langle \Bigl(E - \tau\delta_0(x)\Bigr) u(x), u(x) \Bigr\rangle \leq 0.
\end{equation}
Calculating the scalar product and applying the normalization of the wave function $||u(x)||^2 = \sum_{a\in\Gamma} u^2(a) d(a) = 1$ we obtain
\begin{equation}
E - \tau u^2(0) \geq 0.
\end{equation}
Taking into account that the transparency can take both positive (potential barrier) and negative (potential well) values we obtain inequalities for allowed energies
\begin{eqnarray}
\frac{E}{\tau} &\geq& u^2(0), \;\;\tau\geq 0, \\
\frac{E}{\tau} &\leq& u^2(0), \;\;\tau < 0.
\end{eqnarray}
The first inequality implies
\begin{equation}
E \geq 0, \;\; \tau\geq 0,
\end{equation}
which means that the negative energies are not allowed in the case of a potential barrier.

In order to understand better the second inequality, let us take a look to the following inequality,
\begin{equation}
u^2(0) d(0) \geq \sum\limits_{a \in \Gamma} u^2(a) d(a) = ||u(x)||^2 = 1.
\end{equation}
Implying this inequality results in
\begin{equation}
|E| \leq \frac{2|\tau|}{d(0)}, \;\; \tau < 0.
\end{equation}
In this way, we see that the value of transparency of the potential well establishes an upper bound for the allowed negative energies. This bound is valid not for the one-dimensional potential only, but for the multidimensional ones as well.

Let us notice that this bound reveals a possibility to fix always an (infinitesimal) transparency $\tau$ close to $d(0)$, so that there exists a solution of the corresponding Schr\"{o}dinger equation with finite energy. For example, we can fix $\tau=d(0)$. Notice that, in this case, the potential $\tau \delta_0$ is equal to the ultrafunction indicator $\chi_{0}$, that in the standard approach would not differ from $0$.\par

\vspace*{5.5mm}
\subsubsection{Splitting of the ultrafunction solution}

Let us discuss a particular aspect of non-Archimedean mathematics in general, and ultrafunctions in particular briefly. As we have highlighted above, the ultrafunction solution of Eq.~(\ref{thisone}) has the form of the sum of the ``extended'' standard solution $\psi^{\circ}$ and the non-Archimedean contribution $\varphi$. Now, one might ask whether it is possible for $\varphi$ to be concentrated in the monad of zero or have some similar additional property.

In principle, it is possible if the delta distribution $\delta$ is represented within the ultrafunctions framework differently. In fact, we have used the delta ultrafunction $\delta_{0}=\frac{\chi_{0}}{d(0)}$ to model the delta distribution, which appears in the standard Hamiltonian in Eq.~(\ref{thisone}), because we believe this is a natural approach to follow in our setting. However, there are several ultrafunctions $u$ in $V^{\circ}(\Omega)$ which correspond to $\delta_{0}$ as distributions, so that the scalar product $\langle u, f^{\circ} \rangle=f(0)$ for every $f\in\mathcal{D}(\Omega)$\footnote{This is the notion of ``association'', which is often considered in non-Archimedean mathematics when dealing with distributions, see, e.g., \cite{col85}.}. Of course, by changing the set of test functions, we obtain different conditions. For example, if we use $V^{\circ}$ instead of $\mathcal{D}(\Omega)$, then the unique ``delta'' is precisely the delta ultrafunction. In this sense, there are several different models of the Hamiltonian in Eq.~(\ref{thisone}) that could be constructed in $V^{\circ}(\Omega)$. This gives us a degree of freedom which is not present in the standard approach.

\subsubsection{Comparison of the approaches}

We summarize the discussion of the interpretation of the Schr\"{o}dinger equation with a delta potential within the ultrafunctions approach in the following table.

\setlength{\tabcolsep}{12pt}
\begin{widetext}
\begin{center}
\begin{tabular}{ p{6.5cm}  p{6.5cm} }
   \textbf{Standard approach} & \textbf{Ultrafunction approach} \\
  \hline            
  Analysis of the solutions of the Schr\"{o}dinger equation~(\ref{thisone}) needs an approximation of the delta function by a finite potential (for example, a square barrier/well) and performing a limit of its zero width. & The solutions can be analyzed directly with the use of the delta ultrafunction, which describes an infinite jump concentrated in $0$. However, the approximation procedure can be used as well due to Theorem~\ref{ultraExistence}.  \\
  \hline            
\end{tabular}

\begin{tabular}{ p{6.5cm}  p{6.5cm} }
    Multidimensional delta potentials are difficult to interpret because of the divergence of the corresponding integrals. & There is a unique framework for the delta potentials of any dimension based on hyperfinite spaces and, thus, hyperfinite matrices. Infinite or infinitesimal numbers can be used as parameters of the equation. In this way, a model comes out which better describes the physical phenomenon. \\
  \hline    
  The number of the eigenstates depends on the number of physical dimensions $\mathscr{D}$. In particular, the number of bound states is infinite for $\mathscr{D}\geq 3$. & There always exists the number of eigenstates equal to the hyperfinite dimensionality of the space of ultrafunctions.\\
  \hline            
   The energies of the eigenstates are finite only for a one-dimensional delta potential. & Natural bounds for the energies of the eigenstates can be estimated independently of the dimension of the system. Moreover, it is possible to obtain the bound states with finite energies by fixing suitable transparency. \\
  \hline            
   --- & The splitting of the corresponding ultrafunction can recover the standard content of the solution.
\end{tabular}
\end{center}
\end{widetext}

\section{Conclusions\label{concl}}

In this paper, we have introduced a new approach to quantum mechanics using the advantages of non-Archimedean mathematics. It is based on ultrafunctions, non-Archimedean generalized functions defined on a hyperreal field of Euclidean numbers $\mathbb{E} = \mathbb{R}^*$, which is a natural extension of the field of real numbers. Euclidean numbers have been chosen because they allow constructing a non-Archimedean framework through a simple notion of $\Lambda$-limit, which gives an advantage concerning other nonstandard approaches to quantum mechanics~\cite{Farrukh, Albeverio1979, Bagarello, AlbeverioBook, Raab}.

The ultrafunctions approach proposes a new set of the axioms of quantum mechanics. Even though the new axioms are obviously related with the standard Dirac -- von Neumann formulation of quantum mechanics, ultrafunctions make it possible to build a framework which is closer to the matrix approach. Moreover, the presence of infinite and infinitesimal elements allows constructing new simplified models of physical problems. In particular, this framework is based on Hermitian operators, so that unbounded self-adjoint operators are not more needed. Furthermore, the difference between continuous and discrete spectra is not more present.

These aspects were illustrated by a comparison between standard and ultrafunction solutions of a Schr\"{o}dinger equation with a delta potential. We have shown that the ultrafunctions approach provides a more straightforward way to solve the Schr\"{o}dinger equation by unifying the cases of a potential well and a potential barrier and providing the extreme cases of infinite and infinitesimal transparency of the potential. Moreover, it gives a connection between energy and transparency, providing a bound for the allowed energies of the quantum system.

\begin{acknowledgements}
L.~Luperi~Baglini and K.~Simonov acknowledge financial support by the Austrian Science Fund (FWF): P30821.
\end{acknowledgements}

\bigskip

\bigskip

\end{document}